\documentclass[11pt]{article}
\usepackage{defs}

\title{Complexity of the Graph Homomorphism Problem w.r.t. Degeneracy}
\author{
   Grigorii Braulov
   \thanks{Neapolis University Pafos. Email: \url{braulov2004@gmail.com}}
   \and
   Nikolai Chukhin
   \thanks{JetBrains Research. Email: \url{buyolitsez1951@gmail.com}}
   \and
   Alexander S. Kulikov
   \thanks{JetBrains Research. Email: \url{alexander.s.kulikov@gmail.com}}
   \and
   Ivan Mihajlin
   \thanks{JetBrains Research. Email: \url{ivmihajlin@gmail.com}}
}

\date{}

\begin{document}

\pagenumbering{arabic}
\maketitle

\begin{abstract}
    The graph homomorphism problem~\problemabbr{HOM}~is: given
    an~$n$-vertex source graph~$G$ and an~$h$-vertex target graph~$H$,
    is~there a~mapping from~$V(G)$ to~$V(H)$ that preserves edges?
    A~straightforward brute-force algorithm for~$\problemabbr{HOM}$
    has running time $O(h^n)=O(2^{n\log h})$ and~it~is known that,
    under \cc{ETH},
    there are no~$2^{o(n\log h)}$ algorithms.
    In~recent years,
    less restrictive graph parameters~$p$ have been identified
    that allow one to~solve \problemabbr{HOM} in~time $p(H)^{O(n)}$.
    Examples include treewidth, maximum degree, and track number.
    These algorithms are faster than $O(h^n)$ and allow one
    to~solve~\problemabbr{HOM} in~plain-exponential time~$2^{O(n)}$ in~the special case when $p(H)$~is bounded.
    On~the other hand, it~is known that the chromatic number parameter
    is~too small: under \cc{ETH},
    \problemabbr{HOM} cannot be~solved in~time $\chi(H)^{O(n)}$.

    We~study the complexity of~\problemabbr{HOM} in~terms of~the \emph{degeneracy} of~$H$.
    This~is perhaps the most natural unresolved graph parameter between the known algorithmic and hardness regimes: on~the one hand, each of~bounded treewidth, bounded maximum degree, and bounded track number implies bounded degeneracy; on~the other hand, bounded degeneracy implies bounded chromatic number.
    Our results show that, at~the same time, the influence
    of~degeneracy of~$H$ on~the complexity of~$\problemabbr{HOM}$
    differs significantly from that
    of~the previously studied parameters.
    We~show that, under \cc{ETH}, there is~no $2^{o(\operatorname{degen}(H)n)}$ algorithm for~any value of $\operatorname{degen}(H)$ as a~function of~$n$.
    We also~show that bounded degeneracy alone does not make target size benign: even targets with $\operatorname{degen}(H) \le 2$ and quasi-polynomial size force $n^{\Omega(n)}$-scale hardness.
	Finally, we~introduce a~no-compression barrier that explains why the~known fine-grained lower bounds for~sparse $2$-$\problemabbr{CSP}$ are~not tight under~$\cc{ETH}$.
	Moreover, it~shows that substantially stronger lower bounds for~polynomial-target degeneracy are~unlikely to~follow from~standard reductions from~sparse $3$-$\problemabbr{SAT}$.
\end{abstract}

\clearpage
\tableofcontents
\clearpage
\section{Graph Homomorphism and Target Parameters}
A~graph homomorphism $G \to H$ is a~map from the vertices of a~source graph~$G$ to the vertices of a~target graph~$H$ that preserves edges.
Homomorphisms to a~fixed target encode assignments of local states, while homomorphisms from a~small pattern encode global witnesses such as~cliques or~subgraph counts~\cite{Lovasz12,HN26}.
The corresponding decision problem~\problemabbr{HOM}, \problemname{Graph Homomorphism Problem}, asks whether there exists a~homomorphism between two given graphs.
Various \NP{}-hard graph problems are natural special cases of $\problemabbr{HOM}$: for example, $G$~has a~proper $k$-coloring if~and only~if there exists a~homomorphism $G\to K_k$, whereas $G$~contains a~$k$-clique if and only if there is a~homomorphism $K_k\to G$.
On~the other hand, \problemabbr{HOM} with an~$n$-vertex source graph~$G$ and an~$h$-vertex target graph $H$ is a~special case of \problemname{$2$-Constraint Satisfaction Problem} ($2$-\problemabbr{CSP}) with $n$ variables over a~domain of size~$h$.
Thus, $\problemabbr{HOM}$ may be viewed as $\problemabbr{CSP}$ over a~graph template.

A~straightforward brute-force algorithm (that goes through all mappings)
for \problemabbr{HOM} and $2$-\problemabbr{CSP} has running time $O(h^n)=2^{O(n\log h)}$.
Traxler~\cite{Traxler08} showed that, for $2$-\problemabbr{CSP}, this upper bound cannot be~improved substantially:
under the Exponential Time Hypothesis (\cc{ETH}), there~is no~algorithm
solving $2$-\problemabbr{CSP} in~time $h^{o(n)}=2^{o(n\log h)}$.
At~the same time, a~naive upper bound for the \problemname{$h$-Coloring} problem (asking whether a~given graph admits a~proper $h$-coloring) is also $O(h^n)=2^{O(n\log h)}$, though it~can be~solved in~single-exponential time:
the first such algorithm has running time $O(2.45^n)$ and is~due to~Lawler~\cite{Lawler76},
whereas the best currently known upper bound
is~$O^*(2^n)$ due to~Bj{\"o}rklund, Husfeldt, and Koivisto~\cite{BHK09} ($O^*(\cdot)$ suppresses multiplicative factors
that grow polynomially).
Obtaining a~similar single-exponential upper bound for \problemabbr{HOM}
would be~challenging.
As~proved by~\cite{CHKX06},
under \cc{ETH}, for any $\varepsilon>0$ and any $k=O(n^{1-\varepsilon})$, there~is no~algorithm solving \problemname{$k$-Clique}
for $n$-vertex graphs in time $n^{o(k)}$.
Since
\problemname{$k$-Clique} is~essentially $\problemabbr{HOM}(K_k, \cdot)$,
this rules out (under \cc{ETH}) an~upper bound $2^{o(n\log h)}$
in~the regime where $n$~is much smaller than~$h$. \cite{CFGKMPS17}
proved the same lower bound
for \emph{any} regime: under \cc{ETH}, for any $h=h(n)$, there~is
no~algorithm solving \problemabbr{HOM} in~time $2^{o(n\log h)}$.
This holds for the general case when
both graphs $G$~and~$H$ are part of~the input and
are not restricted in any way. At~the same time, in~various applications
of~\problemname{Graph Homomorphism}, one or~both of~these graphs
may come from a~restricted graph family or~even be~fixed.
A~broader discussion of~homomorphism complexity dichotomies is~given in~\Cref{sec:fine-grained-dichotomy}.

In~recent years,
less restrictive graph parameters~$p$ have been identified
that allow one to~solve \problemabbr{HOM} in~time\footnote{Technically,
    some of~the upper bounds mentioned above have an~additional multiplicative term $2^h$ that is~needed
    to~compute a~certificate for the corresponding parameter of~$H$ (such~as coloring or~decomposition). When $h=O(n)$ (as~in many interesting applications of~\problemabbr{HOM}), this term is~absorbed by~$p(H)^{O(n)}$. For this reason and to~make the summary cleaner,
    we~omit the $2^h$ term.}
$p(H)^{O(n)}$:
  treewidth~\cite{FHK07},
  clique-width~\cite{Wahlstrom11},
  bounded complement bandwidth~\cite{Rzazewski14},
  maximum degree~\cite{FGKM15},
  extended clique-width~\cite{BD20},
  track number and persistent majority number~\cite{Carbonnel26}.
These algorithms are faster than $O(h^n)$ and allow one
to~solve~\problemabbr{HOM} in~plain-exponential time~$2^{O(n)}$ in~the special case when $p(H)$~is bounded.
On~the other hand, as~proved by~\cite{FGKM15},
chromatic number is~too small:
under \cc{ETH},
there~is no~$\chi(H)^{O(n)}$ algorithm.
The known lower and upper bounds with respect
to~various target graph parameters as~well as~connections between them
are
summarized in~\Cref{fig:parameter-dominance-map}.
It~is interesting to~note that, for each target graph parameter~$p$
from the left part of~the figure, there is~no $p(H)^{o(n)}$ algorithm,
under~\cc{ETH}. This follows from the general lower bound due to~\cite{CFGKMPS17} via padding (see \Cref{obs:parameter-padding-lower-bound}).
Also, the argument of~\cite{CFGKMPS17} together with~\cite{FGKM15} can be~used to~prove an~optimal \cc{ETH}-based lower bound $n^{\Omega(n)}$ for graphs~$H$ with $\chi(H)=O(1)$ (see \Cref{sec:bounded-chi-poly-target}).

\tikzset{p/.style={rectangle, rounded corners, draw=black, minimum height=6mm}}

\begin{figure}[ht]
    \begin{center}
        \begin{tikzpicture}[label distance=1mm, scale=.9, transform shape]
            \tikzstyle{a}=[->, >=latex]

            \node[p, label=below:\cite{FHK07}] at (0, 0) (tw) {$\operatorname{tw}$};
            \node[p, label=below:\cite{Rzazewski14}] at (1, -2) (bw) {$\operatorname{bw}$};
            \node[p, label=below:\cite{Wahlstrom11}] at (3, 0) (cw) {$\operatorname{cw}$};
            \node[p, label=below:\cite{BD20}] at (5, -2) (ecw) {$\operatorname{ecw}$};
            \node[p, label=below:\cite{Carbonnel26}] at (6, 0) (tn) {$\operatorname{tn}$};
            \node[p, label=below:\cite{Carbonnel26}] at (8, 0) (pmn) {$\operatorname{pmn}$};
            \node[p, label=below:\cite{FGKM15}] at (4, -4) (delta) {$\Delta$};
            \node[p] at (10.5, 0) (degen) {$\operatorname{degen}$};
            \node[p, label=below:\cite{FGKM15}, right] at (13, 0) (chi) { $\chi$};

            \draw[gray, dashed] (9, -5) -- (9, 1);
            \node[above left, text width=32mm] at (9, -5) {easy parameters~$p$:\\ $\problemabbr{HOM}(\cdot, \mathcal{H}) \in 2^{O(n)}$\\
            if $p(\mathcal{H})=O(1)$};

            \draw[gray, dashed] (12, -5) -- (12, 1);
            \node[above right, text width=35mm, draw=none] at (12, -5) {hard (under \cc{ETH}):\\ $\problemabbr{HOM}(\cdot, \mathcal{H}) \not \in n^{o(n)}$\\if $p(\mathcal{H})=O(1)$};

            \foreach \f/\t/\out/\in in {tw/cw/0/180, cw/ecw/-30/90, tw/tn/20/160,
                tn/pmn/0/180, delta/pmn/10/-150,
                bw/cw/90/-150, degen/chi/0/180,
                pmn/degen/0/180
            }
            \path (\f) edge[a, out=\out, in=\in, <-] (\t);
        \end{tikzpicture}
    \end{center}
	\caption{
        Summary of~previously studied target graph parameters
        and their influence on~the complexity of~\problemabbr{HOM}.
        An~arrow from~$p$ to~$q$ is~read as~follows:
        bounded~$q$ bounds~$p$; more formally, there exists a~function~$f$
        such that, for any graph~$H$, $p(H) \le f(q(H))$ (all shown connections between parameters are proven in~\Cref{section:parameters}). In~turn, this implies
        that an~algorithm solving \problemabbr{HOM} in~plain-exponential
        time in~the special case when $p$~is bounded also has plain-exponential time when $q$~is bounded.
    }
    \label{fig:parameter-dominance-map}
\end{figure}
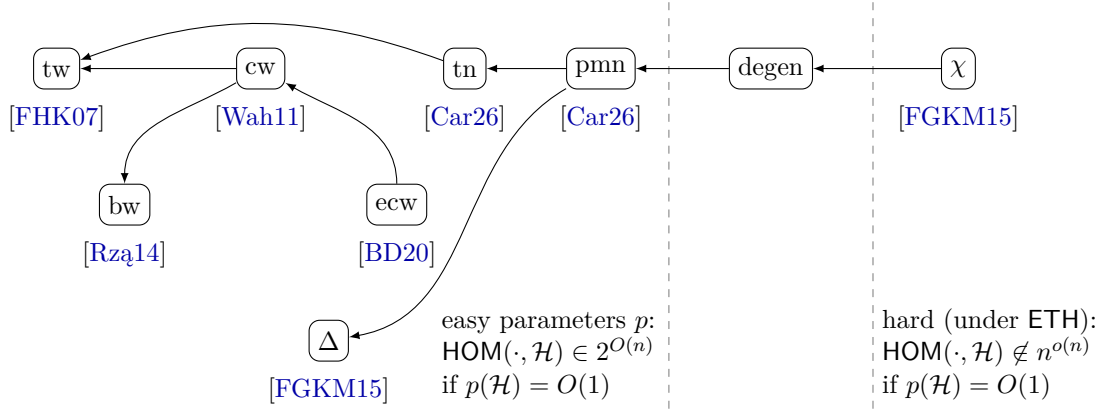

\subsection{Our Results}
We~study the complexity of~\problemabbr{HOM}
in~terms of~target \emph{degeneracy} (the~minimum integer~$d$ such that every nonempty subgraph has a~vertex of~degree at most~$d$), which~is perhaps the most natural unresolved candidate between the known algorithmic and hardness regimes: on~the one hand, each of~bounded treewidth, bounded maximum degree, and bounded track number implies bounded degeneracy; on~the other hand, bounded degeneracy implies bounded chromatic number.

Algorithmically, degeneracy has long served as a~measure of sparse structure~\cite{Akkoyunlu73,MB83,CN85,CCC06,AG09,PRS09,ELS10,PRS12}, and it appears in~tight kernelization bounds for sparse graph problems~\cite{CPPW12,DM12,HW12,CGH17}.
The $\problemabbr{HOM}$ counting problem was studied in~degenerate graphs, but the degeneracy constraint is~usually placed on~the large input graph into which a~fixed pattern is mapped.
Starting with Bressan~\cite{Bressan21}, this line gives near-linear and fine-grained classifications for counting subgraphs, induced subgraphs, homomorphisms, and homomorphic cycles in degenerate graphs~\cite{BPS20,BPS21,BGLSS22,BR21,GLSY23,PS24,PS25,KKP25}.
In~those results, the pattern is fixed and the running time is measured as a~function of~the sparse host, whereas in $\problemabbr{HOM}$ both $G$ and $H$ are part of~the input.
There are also related results for bounded outdegree graphs~\cite{BLR23}, bounded-degeneracy hypergraphs~\cite{PS26}, and recognition by homomorphism counts to $2$-degenerate graphs~\cite{Dvorak10}.

Our results show that the influence
of~degeneracy of~$H$ on~the complexity of~$\problemabbr{HOM}$
differs significantly from that
of~the previously studied parameters.
We~show that, for every nondecreasing unbounded threshold function $D$ satisfying $D(n)=O(n^{1/3})$ and $D(2n)=O(D(n))$, the promise $\operatorname{degen}(H)\le D(n)$ problem admits no $2^{o(D(n)n)}$ algorithm under~\cc{ETH}.

\begin{restatable}[Informal version of \Cref{thm:main-lb-formal}]{theorem}{mainDegLB}
    \label{thm:main-lb}
    Under~\cc{ETH}, for every nondecreasing unbounded function $D\colon\mathbb N\to\mathbb N$ satisfying $D(n)=O(n^{1/3})$ and $D(2n)=O(D(n))$, there~is no algorithm solving $\problemabbr{HOM}$ on~every instance satisfying $\operatorname{degen}(H)\le D(n)$ in~time\footnote{Recall that $O^*(\cdot)$ suppresses factors that grow polynomially in~the input length. For~\problemabbr{HOM}, such factors are of~the form $(n+h)^{O(1)}$.} $O^*\left(2^{o(D(n)n)}\right)$.
\end{restatable}

We then show that the plain-exponential frontier is~already crossed between degeneracy~$1$ and~$2$.
Targets with degeneracy at~most~$1$ are forests and therefore have bounded treewidth, so~they are plain-exponential by~\cite{FHK07}.
In~contrast, we construct hard instances of degeneracy~$2$, provided that~the target is allowed to~have size quasi-polynomial in~$n$.

\begin{restatable}[Informal version of \Cref{thm:constant-deg-lb-formal}]{theorem}{constantDegLB}
    \label{thm:constant-deg-lb}
    Under $\cc{ETH}$, no~algorithm solves $\problemabbr{HOM}$ in~time
    \(O^*(n^{o(n)})\)
    on instances $(G,H)$ satisfying $\operatorname{degen}(H)\le2$ and $h\le 2^{O(\log^2 n)}$.
\end{restatable}

Finally, we~formulate a~new barrier showing that improving the lower bound $2^{o(\operatorname{degen}(H)n)}$ further~to $2^{\omega(\operatorname{degen}(H)n)}$ would be~challenging.
Informally, the reason is~that such a~reduction would have to~be stronger than all known fine-grained reductions from $3$-$\problemabbr{SAT}$. We~make~it formal below.

By the Sparsification Lemma~\cite{IPZ01}, $\cc{ETH}$ can be~viewed
as~hardness of~sparse $3$-\problemabbr{SAT} with $n$~variables and
$O(n)$~clauses.
An~explicit sparse formula has bit length $\Theta(n\log n)$, therefore if $s$~denotes the bit size of~the input, then sparse $3$-\problemabbr{SAT} is $2^{\Omega(s / \log s)}$ hard.
Many fine-grained reductions are known for various problems; however, to~the best of~our knowledge, not a~single one of~them establishes a~lower bound better than $2^{\Omega(s / \log s)}$, where $s$~denotes the input bit size.

A~concrete example of~a~problem exhibiting a~gap between the~known lower and~upper bounds is~the sparse $2$-\problemabbr{CSP} problem.
A~brute-force algorithm for~the $2$-\problemabbr{CSP}
with $n$~variables over a~domain of~size~$h$
has running time $O(h^n)$ and there are no~$h^{o(n)}$ algorithms under~$\cc{ETH}$~\cite{Traxler08}.
However, when the~number of~constraints is~linear, the~best known lower bound is~only $h^{\Omega(n / \log n)}$~\cite{Marx10,KMPS24}.
The~question of~improving the lower bound to~$h^{\Omega(n)}$ was raised in~\cite{KMPS24}.
We~note that, in~order to~close this gap, the~lower bound would have to~establish a~bound stronger than~$2^{\Omega(s / \log s)}$, where~$s$ denotes the~input size in~bits; see~\Cref{subsec:sparse-csp-nocomp}.

\begin{restatable}[No-Compression Hypothesis, informal version]{hypothesis}{noCompressionHyp}\label{hyp:nocomp}
    There~is no~fine-grained reduction establishing a~lower bound
    $2^{\omega(s/\log s)}$ under \cc{ETH}, where $s$~is the bit size
    of~the input.
\end{restatable}

A~weaker version of~this hypothesis was recently considered
by~Kulikov and Mihajlin~\cite{KM24} who proved that an~\cc{SETH}-based
lower bound $2^{\Omega(s)}$ violates \cc{SETH}. We~show that
refuting a~somewhat stronger version of~\Cref{hyp:nocomp} implies new
circuit lower bounds: if~there~is a~uniform generator that converts $O(n^{0.99})$ input bits into CNF formulas over $n$~variables that are~$2^{\Omega(n)}$-hard under \cc{ETH}, then $\P^{\NP}$ has truly superlinear circuit size.

Finally,
we~show that a~reduction showing a~$2^{\omega(\operatorname{degen}(H)n)}$ lower bound for \problemabbr{HOM} under \cc{ETH} would refute~\Cref{hyp:nocomp}.

\begin{restatable}[Informal version of~\Cref{thm:unified-nocomp-formal}]{theorem}{dictNoComp}
    \label{thm:unified-nocomp}
    Under \Cref{hyp:nocomp},
    no~reduction
    can establish a~$2^{\omega(\operatorname{degen}(H)n)}$ lower bound for \problemabbr{HOM} under \cc{ETH}.
\end{restatable}

\paragraph{Techniques.}
The lower bound of~\Cref{thm:main-lb} is~a~reduction from the \problemname{$3$-Coloring} problem on~graphs of~maximum degree~$4$.
We~partition $N$ vertices of~the instance into vertex-buckets of~size~$r$, so~that the~source of~the produced $\problemabbr{HOM}$ instance has only $O(N/r)$ vertices, while the~target must supply a~state for~every coloring of~a~bucket.
The~difficulty is~to~check the~constraints \emph{between} buckets without making these states densely connected; we~do~this by~assigning $O(r)$ labels to~the~buckets via~an~equitable coloring, so~that every target vertex keeps degree~$O(r)$ and hence $\operatorname{degen}(H)=O(r)$, while $\operatorname{degen}(H)\cdot n$ stays linear in~$N$.
The~list constraints of~the intermediate list-homomorphism instance are then removed by~attaching a~small rigid frame to~both graphs.
The~main difference between \Cref{thm:main-lb} and~the previous lower bounds for~$\problemabbr{HOM}$~\cite{FGKM15,CFGKMPS17} is~that we~bucket not only the~vertices but~also the~constraints between them.
This approach allows us~to keep the~target degenerate whereas making the~testing problem more difficult.

For~\Cref{thm:constant-deg-lb}, we~start from the~$\cc{ETH}$-hardness of~sparse bounded-alphabet $2$-\problemabbr{CSP}~\cite{KMPS24}, encode it into $\problemabbr{HOM}$ while keeping \emph{both} graphs sparse, and then subdivide every edge into a~path of~length three; subdivision preserves homomorphisms while making the~target $2$-degenerate, and sparsity keeps the~size growth quasi-polynomial.

The~barrier of~\Cref{thm:unified-nocomp} rests on the following idea: in~a~$D$-degenerate target, the~set of~images still available to~a~source vertex under some~partial homomorphism always admits a~record of~$O(D\log n)$ bits.
This reduces $\problemabbr{HOM}$ to~a~succinct homomorphism problem whose instances can be encoded using a~few bits, so~a~reduction proving a~$2^{\omega(\operatorname{degen}(H)n)}$ lower bound would compress sparse $3$-$\problemabbr{SAT}$ beyond the~$s/\log s$ scale, contradicting \Cref{hyp:nocomp}.

\paragraph{Organization.}
\Cref{sec:preliminaries} fixes notation and assumptions.
\Cref{sec:linear-lb} proves the linear degeneracy lower bound.
\Cref{sec:constant-degenerate-nlogn} proves the constant-degeneracy quasi-polynomial-target lower bound.
\Cref{sec:no-compression} formalizes the no-compression barrier and proves its $\problemabbr{HOM}$ and sparse $\problemabbr{CSP}$ consequences.
\Cref{sec:generated-hard-witnesses} proves the circuit lower-bound consequence.
\Cref{sec:conclusion} summarizes the resulting target-side picture and lists open problems.
\Cref{sec:bounded-chi-poly-target} gives the polynomial-target bounded-chromatic lower bound.

\section{Preliminaries}
\label{sec:preliminaries}

All logarithms in this paper are taken base $2$.
For a~positive integer $q$, we write $[q]=\{1,\ldots,q\}$.
For a~nonnegative integer $q$, by $\operatorname{bin}(q)$ we mean a~binary encoding of $q$.
For a~partial function $f$, we write $\operatorname{dom}(f)$ for its domain.
All graphs are finite, undirected, and loopless unless explicitly stated otherwise.
For a~graph $G$, we write $V(G)$ and $E(G)$ for its vertex and edge sets, $N_G(v)$ for the open neighborhood of~$v$, $\Delta(G)$ for its maximum degree, and $\overline G$ for its complement.
For a~set $S\subseteq V(G)$, the graph $G[S]$ denotes the subgraph induced by~$S$.

Throughout the paper, $G$ denotes the source graph and $H$ denotes the target graph in a~homomorphism instance.
We write $n=|V(G)|$ and $h=|V(H)|$.
In reductions from a source problem, $N$ denotes the size parameter of the original instance.

A~hypergraph $\mathcal A=(C,\mathcal E)$ consists of a~finite vertex set~$C$ and a~family $\mathcal E$ of subsets of~$C$, called hyperedges.
It is $d$-uniform if every hyperedge has size exactly~$d$.
An~automorphism of~$\mathcal A$ is a~permutation of~$C$ that preserves the hyperedge family, and $\operatorname{Aut}(\mathcal A)$ denotes the automorphism group.

\subsection{Homomorphism Problems}
A~homomorphism from a~graph $G$ to a~graph $H$ is a~map $\phi\colon V(G)\to V(H)$ such that $uv\in E(G)$ implies $\phi(u)\phi(v)\in E(H)$.
We write $G\to H$ if such a~map exists.
The two-input decision problem is denoted by $\problemabbr{HOM}$.
For graph classes $\mathcal G$ and $\mathcal H$, the notation $\problemabbr{HOM}(\mathcal G,\mathcal H)$ denotes the restriction to instances $(G,H)$ with $G\in\mathcal G$ and $H\in\mathcal H$.
The symbol ``$\cdot$'' means that the corresponding side is unrestricted; thus $\problemabbr{HOM}(\cdot,H)$ is the fixed-target $H$-coloring problem and $\problemabbr{HOM}(\cdot,\mathcal H)$ is the restricted variable-target problem with targets from~$\mathcal H$.

In~the \problemname{List Graph Homomorphism} problem, \problemabbr{LHOM}, the input is a~pair of~graphs~$(G,H)$
together
with a~list $L(v)\subseteq V(H)$ for every $v\in V(G)$.
The goal~is to~check
whether there is a~homomorphism $\phi \colon G\to H$ satisfying $\phi(v)\in L(v)$ for every source vertex~$v$.

\subsection{Graph Parameters}
We use the following basic graph parameters.
\begin{itemize}
	\item The chromatic number $\chi(G)$ is the least number of colors in~a~proper vertex coloring of~$G$.

	\item The maximum degree $\Delta(G)$ is the maximum degree of a~vertex of~$G$.

	\item The degeneracy $\operatorname{degen}(G)$ is the least integer $d$ such that every non-empty subgraph of~$G$ has a~vertex of degree at most~$d$.
		Equivalently, $G$ has an~ordering in which each vertex has at most~$d$ later neighbors.
\end{itemize}

The remaining target-side parameters appearing in Figure~\ref{fig:parameter-dominance-map} are defined formally in Appendix~\ref{section:parameters}, where we also state parameter implications.

\subsection{Constraint Satisfaction}
A~finite-domain $\problemabbr{CSP}$ instance consists of variables, a~finite domain, and constraints restricting tuples of variables.
In this paper, the source problems are binary CSPs.
For $h\ge2$, a~$(2,h)$-$\problemabbr{CSP}$ instance has variables taking values in~$[h]$, unary constraints $R_x\subseteq[h]$, and binary constraints $R_{xy}\subseteq[h]^2$ on ordered pairs of distinct variables.
An~assignment satisfies the instance if it satisfies every unary and binary relation.
The constraint graph has one vertex for each variable and one edge for each binary constraint, with no orientation.
More generally, a~$(k,h)$-$\problemabbr{CSP}$ instance has $n$ variables over~$[h]$ and constraints of~arity at most~$k$.
A~\emph{sparse} $(k,h)$-$\problemabbr{CSP}$ instance is one with at most~$O(n)$ constraints.
In the encoding of $\problemabbr{CSP}$, every constraint encodes its scope and the truth table of its relation.

\subsection{Boolean Circuits}
A~Boolean circuit $C$ over variables $x_1, \dotsc, x_n$ is a~directed acyclic graph with nodes of in-degree zero or two.
The in-degree zero nodes are labeled by variables $x_i$ and constants $0$ or $1$, whereas the in-degree two nodes are labeled by binary Boolean functions.
The only out-degree zero node is the output of the circuit.
A~circuit $C$ computes a~Boolean function $f  \colon \{ 0, 1 \}^{n} \to \{ 0, 1 \}$ in a~natural way.

For a~language $L\subseteq\{0,1\}^*$, a~circuit $C$ with $n$ input variables computes $L$ on inputs of length~$n$ if, for every $x\in\{0,1\}^n$, $C(x) = 1$ exactly when $x\in L$.
For a~language $L$, let $\cc{CC}_L[n]$ denote the minimum size of a~Boolean circuit computing $L$ on all inputs of length~$n$.
For a~function $s\colon\mathbb N\to\mathbb N$, write
\[L\in\cc{SIZE}[s(n)]\]
if for every sufficiently large~$n$,
\[\cc{CC}_L[n]\le s(n).\]

Proving lower bounds on~the circuit size of~explicit functions is~a~long-standing challenge.
A~counting argument due to~Shannon~\cite{Shannon49} shows that almost all Boolean functions on~$n$ variables require circuits of~size $\Omega(2^n/n)$.
However, the~best known lower bound for an~explicit function\footnote{By an~explicit function, we~mean a~function computable in~$\P^{\NP}$.} is~$3.1n-o(n)$, due to~Li and~Yang~\cite{LiY22}.

For a~string $a\in\{0,1\}^N$, let $\cc{CC}[a]$ denote the minimum size of a~circuit on $k=\ceil{\log_2 N}$ input bits whose first $N$ truth-table values are the bits of~$a$.
The remaining $2^k-N$ truth-table values are ignored.

\subsection{SAT Complexity Hypotheses}
A~$k$-$\cc{CNF}$ formula is a~conjunction of clauses of width at most~$k$, and $k$-$\problemabbr{SAT}$ asks whether such a~formula is satisfiable.
A~sparse $3$-$\problemabbr{SAT}$ instance on $N$ variables has $O(N)$ clauses.

\begin{hypothesis}[Exponential Time Hypothesis, $\cc{ETH}$~\cite{IPZ01}]
	There is no algorithm solving $3$-$\problemabbr{SAT}$ on formulas with $N$ variables in time $2^{o(N)}$.
\end{hypothesis}

By the Sparsification Lemma of Impagliazzo, Paturi, and Zane~\cite{IPZ01}, $\cc{ETH}$ is equivalent to the statement that sparse $3$-$\problemabbr{SAT}$ has no $2^{o(N)}$-time algorithm.

\begin{hypothesis}[Strong Exponential Time Hypothesis, $\cc{SETH}$~\cite{IP01}]
	For every $\varepsilon>0$, there exists an~integer $k\ge3$ such that $k$-$\problemabbr{SAT}$ on formulas with $N$ variables cannot be solved in time $2^{(1-\varepsilon)N}\operatorname{poly}(N)$.
\end{hypothesis}

\begin{hypothesis}[Counting Strong Exponential Time Hypothesis, $\#\cc{SETH}$~\cite{DHMTW14,CM16}]
	For every $\varepsilon>0$, there exists an~integer $k\ge3$ such that counting the satisfying assignments of $k$-$\cc{CNF}$ formulas with $N$ variables cannot be done in time $2^{(1-\varepsilon)N}\operatorname{poly}(N)$.
\end{hypothesis}

The following bounded-degree coloring consequence of $\cc{ETH}$ is due to Kochol, Lozin, and Randerath~\cite{KLR03}.
\begin{theorem}\label{thm:bdcolor}
	Assuming $\cc{ETH}$, there exists a~constant $\delta>0$ such that \problemname{$3$-Coloring}
	on~$N$-vertex graphs of~maximum degree at most~$4$ cannot be solved in time $2^{\delta N}$.
\end{theorem}

Throughout the~paper, we~use the~notion of a~fine-grained reduction; for~its formal definition and~an~overview of~results in~fine-grained complexity, see~\cite{Vassilevska18}.

\section{Linear Dependence on Degeneracy}
\label{sec:linear-lb}

This section reduces \problemname{$3$-Coloring} on bounded-degree graphs to \problemabbr{HOM}, on targets whose degeneracy is small.
The proof has two separate parts.
First, we reduce \problemname{$3$-Coloring} to an~\problemabbr{LHOM} instance whose target already has small degeneracy.
Then, we encode the list constraints into ordinary graph homomorphism.
The construction is arranged so~that the target degeneracy, multiplied by the number of source vertices, stays linear in~the size of the original instance.
Consequently, an~algorithm whose running-time exponent were sublinear in~this product would solve bounded-degree \problemname{$3$-Coloring} faster than $\cc{ETH}$ permits.

\begin{restatable}[Formal version of \Cref{thm:main-lb}]{theorem}{mainDegLBFormal}
	\label{thm:main-lb-formal}
	Unless $\cc{ETH}$ fails, for every nondecreasing unbounded function $D\colon\mathbb N\to\mathbb N$ satisfying $D(n)=O(n^{1/3})$ and $D(2n)=O(D(n))$, there is no algorithm that solves $\problemabbr{HOM}$ on every instance $(G,H)$ satisfying $\operatorname{degen}(H)\le D(n)$ in~time
	\[O^*\left(2^{o(D(n)n)}\right).\]

	More precisely, for every sufficiently large integer~$r$ and all sufficiently large~$N$, every $N$-vertex graph $Q$ of maximum degree at most~$4$ can be transformed into an~instance $(G_r,H_r)$ of $\problemabbr{HOM}$, in~time polynomial in~$|Q|+|H_r|$, with
		\[|V(G_r)|=O(N/r+r^2),\qquad |V(H_r)|=2^{O(r^2)},\]
		and with
	\[\operatorname{degen}(H_r)=O(r),\]
	such that $Q$ is $3$-colorable if and only if there is a~homomorphism $G_r\to H_r$.
\end{restatable}

\subsection{Proof Overview}
We take a~graph~$Q$ on~$N$ vertices with maximum degree at most~$4$, and partition its vertices into $\ceil{N/r}$ \emph{vertex-buckets} of~size at most~$r$; we fold each vertex-bucket into a~single vertex of a~list-homomorphism source.

The obstacle is that folding $r$ vertices into one source vertex forces the target to supply, for each group, a~vertex for every possible color pattern in~$[3]^r$.
The target may therefore have $2^{O(r^2)}$ vertices, and the whole difficulty is to verify the $3$-coloring constraints \emph{between} groups without making these vertices densely connected.
Classical reductions do exactly that: they attach to each state a~consistency gadget that lists all compatible states of its neighbors, and since a~single group can interact with $\Theta(r)$ others, the resulting states have degree exponential in~$r$.

We avoid this by further partitioning $E(Q)$ into \emph{edge-buckets} of at most $r$ edges each.
The partition is chosen so that the vertex-buckets meeting in~one edge-bucket have distinct \emph{labels}; hence one target state can refer to them by label without ambiguity.

The list target has two kinds of state vertices, plus a~test board.
A~\emph{bucket-state} vertex stores one color string $c\in[3]^r$ for one labeled bucket.
An~\emph{edge-state} vertex stores only the local data for one edge-bucket: the color strings of the participating labels and, for each edge in~that edge-bucket, the two bucket positions that it connects.
The global \emph{test board} contains small test vertices that check one edge constraint at a~time.

The point of edge splitting is locality in checking. Since an~edge-state vertex is adjacent only to the bucket-states of its own participating labels and to the test vertices for its own edge-bucket, its degree is $O(r)$ even though there may be $2^{O(r^2)}$ edge-states in~total.

After constructing this sparse list-homomorphism instance, we remove the lists.
We add a~small rigid \emph{frame} to both graphs and connect each source vertex to the frame according to the target vertices it is allowed to use.
Rigidity forces every homomorphism to respect these connections, so the lists are simulated inside ordinary $\problemabbr{HOM}$.

The rest of the section follows this order.
\Cref{subsec:bucketing} builds the buckets and labels.
\Cref{subsec:list-instance} defines the \problemabbr{LHOM} instance and proves that it is equivalent to \problemname{$3$-Coloring}.
\Cref{subsec:remove-lists} replaces lists by the rigid frame.
\Cref{subsec:parameters-eth} finishes the parameter bounds and derives the $\cc{ETH}$ contradiction.

\subsection{Bucket Decomposition}
\label{subsec:bucketing}

Let $Q$ be an $N$-vertex graph of maximum degree at most $4$.
Fix an~integer $r\ge2$.
We construct an~\problemabbr{LHOM} instance $(\Lambda_r,\Gamma_r,\mathcal L_r)$, where $\Lambda_r$ is the source graph, $\Gamma_r$ is the target graph, and $\mathcal L_r$ is the list assignment.
Thus, for each source vertex $v\in V(\Lambda_r)$, the set $\mathcal L_r(v)\subseteq V(\Gamma_r)$ consists of the allowed images of~$v$.
Later, we convert it into an~ordinary homomorphism instance~$(G_r,H_r)$.

The reduction groups the vertices and edges of~$Q$ into small blocks.
A~\emph{vertex-bucket} is a~set of at most $r$ vertices of~$Q$; one source vertex of $\Lambda_r$ will be responsible for the colors of an~entire vertex-bucket.
An~\emph{edge-bucket} is a~set of at most $r$ edges of~$Q$; one source vertex of $\Lambda_r$ will be responsible for checking all the $3$-coloring constraints inside one edge-bucket.
We say that a~vertex-bucket \emph{participates} in~an~edge-bucket if it contains an~endpoint of some edge of that edge-bucket.

A~vertex of $\Lambda_r$ that handles an~edge-bucket must know the buckets whose colors it compares, so we give every vertex-bucket a~\emph{label} from a~small set $[L]$ with $L=O(r)$, and we arrange that the buckets participating in~any single edge-bucket carry pairwise distinct labels.
The next lemma produces such a~bucketing and labeling.

\begin{lemma}\label{lem:bucket-labels}
	For any integer $r\ge2$ and all sufficiently large $N$, any graph $Q$ on $N$ vertices with $\Delta(Q) \le 4$ admits:
	\begin{itemize}
		\item a partition of $V(Q)$ into vertex-buckets $B_1,\ldots,B_k$ of size at most
		$r$;
		\item a label map $\lambda\colon\{B_1,\ldots,B_k\}\to[L]$ with $L\coloneqq4r+1$;
		\item a partition of $E(Q)$ into edge-buckets $E_1,\ldots,E_s$ with
		$s=O(N/r+r)$ and $|E_j|\le r$,
	\end{itemize}
	such that, inside each edge-bucket, no label occurs on two different participating
	vertex-buckets.
	Moreover, finding this partition and labeling can be done in polynomial time.
\end{lemma}

\begin{proof}
	Partition $V(Q)$ arbitrarily into buckets of size at most $r$, and let $k=\ceil{N/r}$.
	Let $\mathcal B$ be the bucket graph: two buckets are adjacent if some edge of $Q$ has
	one endpoint in each. Since each bucket has at most $r$ vertices and $Q$ has maximum
	degree at most $4$, we have $\Delta(\mathcal B)\le 4r$. By the
	Hajnal--Szemer{\'e}di equitable coloring theorem~\cite{HS70}, $\mathcal B$
	has an equitable proper coloring with $L=4r+1$ colors; denote it by $\lambda\colon \{ B_1, \ldots, B_k \} \to [L]$.\footnote{An~equitable coloring is a~proper coloring such that the sizes of~any two color classes differ by~at~most one.}
	Each label is used by $O(k/L+1)=O(N/r^2+1)$ buckets.
	Note that this coloring can be found in polynomial time, see e.g.~\cite{KKMS10}.

	Create a multigraph $M$ on vertex set $[L]$. Each edge of $Q$ whose endpoint
	buckets have labels $a$ and $b$ contributes an edge $ab$ to $M$; an edge internal to
	one bucket contributes a loop. Since $\lambda$ is proper on $\mathcal B$, loops come
	only from edges internal to a bucket. For a fixed label $a$, the total number of
	edges incident with buckets of label $a$ is at most
	$4r\cdot O(N/r^2+1)=O(N/r+r)$.
	The line graph of this multigraph has maximum degree $O(N/r+r)$.
	A~greedy coloring of this line graph therefore splits the edges of~$Q$ into $O(N/r+r)$ classes, where each class is a~matching in~$M$.

	These matching classes may still contain more than~$r$ edges, so we split every class arbitrarily into chunks of size at most~$r$ and declare the chunks to be the edge-buckets.
	If the matching classes are $F_1,\ldots,F_t$, where $t=O(N/r+r)$, then the number of chunks is at most
	\[\sum_{i=1}^t \ceil{|F_i|/r}\le |E(Q)|/r+t=O(N/r)+O(N/r+r),\]
	since $|E(Q)|\le 2N$.
	Thus the total number of edge-buckets is $O(N/r+r)$, and each edge-bucket contains at most~$r$ edges.

	It remains to check the distinct-label property.
	Fix an~edge-bucket $E_j \subseteq E(Q)$.
	By construction, all edges in~$E_j$ came from one matching class in~$M$.
	For an~edge $e\in E_j$, let $\mu(e)$ be the edge of~$M$ produced by the labels of the endpoint buckets of~$e$.
	If $e$ is internal to a~single vertex-bucket $B$, then $\mu(e)$ is a~loop at $\lambda(B)$ and $e$ contributes only the participating bucket~$B$.
	If $e$ has endpoints in~two different buckets $B$ and $B'$, then $B$ and $B'$ are adjacent in~$\mathcal B$; since $\lambda$ is a~proper coloring of~$\mathcal B$, we have $\lambda(B)\ne\lambda(B')$.

	Now take two distinct edges $e,f\in E_j$.
	The edges $\mu(e)$ and $\mu(f)$ are distinct edges of one matching in~$M$, so they have no common incident label.
	Therefore no label used by a~vertex-bucket participating in~$e$ can also be used by a~vertex-bucket participating in~$f$.
	So, the labels on all vertex-buckets participating in~$E_j$ are pairwise distinct.
\end{proof}

For the rest of the construction, we pad each vertex-bucket $B_i$ with dummy vertices to make it of size exactly $r$.

\subsection{List-Homomorphism Instance}
\label{subsec:list-instance}

\paragraph{The target.}

Fix a~bucketing and labeling from~\Cref{lem:bucket-labels} for the rest of the section.
We define the list target graph $\Gamma_r$.
Its vertices fall into three types.
A~\emph{bucket-state} vertex stores one color string for one labeled vertex-bucket.
An~\emph{edge-state} vertex stores the color strings of the labels participating in~one edge-bucket, together with a~partial map recording which two bucket-coordinates each edge compares.
The \emph{test board} is a~family of test vertices; each one checks a~declared endpoint of an~edge, or rules out one monochromatic edge.

The bucket-state set is
\[
        S_B=\{U_{\ell,c}\colon\ell\in [L],\ c\in[3]^r\}.
\]
A~bucket-state represents a~coloring string for one bucket of label $\ell$.

The edge-state set is
\[
S_E := \left\{ W_{Z,c,\Pi}  \colon
\begin{array}{l}
	\varnothing\ne Z\subseteq[L],\ |Z|\le 2r,\\
	c=(c_\ell)_{\ell\in Z},\ c_\ell\in[3]^r\text{ for every }\ell\in Z,\\
\Pi \colon \operatorname{dom}(\Pi)\subseteq[r]\to Z\times[r]\times Z\times[r] \end{array}
\right\}.
\]
That is, $\Pi$ is a~partial function from $[r]$ to $Z \times [r] \times Z \times [r]$.
If $\Pi(p)=(a,x,b,y)$, then the edge numbered $p$ compares coordinate~$x$ of label~$a$ with coordinate~$y$ of label~$b$; the test board, defined next, verifies this claim and the resulting color condition.
We make $\Pi$ a~partial function so that an~edge-state compares only the edges of its own edge-bucket.

The test board is the union of the following families:
\[
\begin{aligned}
        T_A&=\{A_{p,a}\colon p\in[r],\ a\in[L]\},&
        T_B&=\{B_{p,b}\colon p\in[r],\ b\in[L]\},\\
        T_X&=\{X_{p,x}\colon p,x\in[r]\},&
        T_Y&=\{Y_{p,y}\colon p,y\in[r]\},\\
        T_C&=\{C_{p,\gamma}\colon p\in[r],\ \gamma\in[3]\}.
\end{aligned}
\]
Let $T=T_A\cup T_B\cup T_X\cup T_Y\cup T_C$.

The list target graph has vertex set
\[
        V(\Gamma_r)=S_B\sqcup S_E\sqcup T.
\]

Its edges are defined as follows.

\begin{itemize}
	\item State edges.
	For $U_{\ell,c}\in S_B$ and $W_{Z,d,\Pi}\in S_E$, add the edge
	\[
			U_{\ell,c}W_{Z,d,\Pi}
	\]
	if and only if
	\[
			\ell\in Z\quad\text{and}\quad d_\ell=c.
	\]

	\item Test edges.
	For $W=W_{Z,c,\Pi}\in S_E$ and $\Pi(p)=(a,x,b,y)$, add the metadata edges
	\[
			WA_{p,a},\qquad WX_{p,x},\qquad WB_{p,b},\qquad WY_{p,y}.
	\]
	For $\gamma\in[3]$, add the color-check edge $WC_{p,\gamma}$ if and only if
	\[
			\Pi(p)=(a,x,b,y)\quad\text{and}\quad \neg(c_a[x]=\gamma=c_b[y]).
	\]
	The expression is evaluated only when $\Pi(p)$ is defined, which already implies $a,b\in Z$.
	If $a=b$, the same string $c_a$ is used in both positions.
	See~\Cref{fig:edge-index-target}.
\end{itemize}

\begin{figure}[ht]
\centering
\begin{tikzpicture}[yscale=.8]
	\tikzstyle{a} = [rectangle, minimum width=13mm, rounded corners, draw]
	\node[a] (t) at (0, 0) {$W_{Z,c,\Pi}$};
	\node[a] (a) at (3, 2) {$A_{p,a}$};
	\node[a] (x) at (3, 1) {$X_{p,x}$};
	\node[a] (b) at (3, 0) {$B_{p,b}$};
	\node[a] (y) at (3, -1) {$Y_{p,y}$};
	\node[a] (c) at (3, -2) {$C_{p,\gamma}$};
	\node[a] (ua) at (-3, 0.5) {$U_{a, c_{a}}$};
	\node[a] (ub) at (-3, -0.5) {$U_{b,c_b}$};

	\node at (0, -2) {$\Pi(p)=(a,x,b,y)$};
	\node at (-1.5, 0.6) {$a \in Z$};

	\foreach \n in {a, x, b, y, c, ua, ub}
		\draw (t) -- (\n);
\end{tikzpicture}

\caption{One edge-state stores only the tests used by its edge-bucket.
The test forces the tuple $\Pi(p)$, and the color-check vertex $C_{p,\gamma}$ forbids the monochromatic case.}
\label{fig:edge-index-target}
\end{figure}
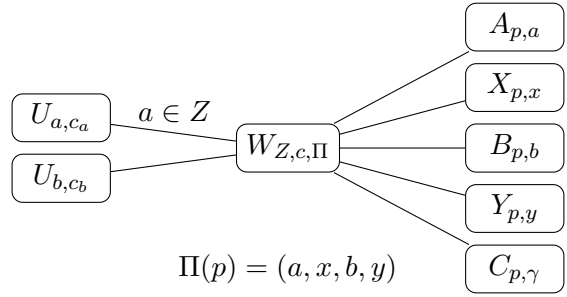

\begin{lemma}\label{lem:target-size}
\[
        |V(\Gamma_r)| = 2^{O(r^2)}.
\]
\end{lemma}

\begin{proof}
	We have
	\[
			|S_B|=L3^r=\Theta(r\,3^r)
	\]
	and
	\[
			|T|=2Lr+2r^2+3r=O(r^2).
	\]
	For the edge-states, for $\mu=|Z|\le 2r$, first choose~$Z$, then choose $\mu$ color strings in $[3]^r$, choose a~domain for~$\Pi$, and finally choose a~value in $Z\times[r]\times Z\times[r]$ for each edge index in that domain.
	Hence
	\[
	|S_E|
	\le
	\sum_{\mu=1}^{2r}
	\binom{L}{\mu}(3^r)^\mu
	\sum_{q=0}^{r}\binom rq(\mu^2r^2)^q.
	\]
	Since $L=\Theta(r)$ and $\mu\le 2r$, the logarithm of this expression is $O(r^2)$: the color-string factor contributes $O(r^2)$ and the remaining combinatorial factors contribute only $O(r\log r)$.
	Thus $|S_E|\le 2^{O(r^2)}$, and therefore $|V(\Gamma_r)| = 2^{O(r^2)}$.
\end{proof}

\paragraph{The source.}

The source graph records the bucket structure of~$Q$: it has a~bucket vertex for each vertex-bucket, to be sent to a~bucket-state, and an~edge-bucket vertex for each edge-bucket, to be sent to an~edge-state.
A~single copy of the test board is shared by all edge-buckets.

The graph $\Lambda_r$ contains:
\begin{itemize}
    \item a copy $T'$ of the entire test board $T$;
    \item one bucket vertex $u_i$ for each vertex-bucket $B_i$;
    \item one edge-bucket vertex $w_j$ for each edge-bucket $E_j$.
\end{itemize}

Add the following edges.
First, join $w_j$ to $u_i$ if and only if $B_i$ participates in $E_j$.
Second, index the edges in each edge-bucket by distinct values $p\in[|E_j|]\subseteq[r]$ and orient them arbitrarily.
If the edge with index $p$ is oriented from coordinate $x$ in bucket $B_\alpha$ to coordinate $y$ in bucket $B_\beta$, put $a=\lambda(B_\alpha)$ and $b=\lambda(B_\beta)$, and add the following edges:
\[
        w_jA'_{p,a},\quad w_jX'_{p,x},\quad w_jB'_{p,b},\quad w_jY'_{p,y},
\]
and the three edges $w_jC'_{p,\gamma}$, for $\gamma\in[3]$.
The same global copied test vertex may be used by many edge-buckets.

\begin{lemma}\label{lem:list-source-size}
\[|V(\Lambda_r)|=O(N/r+r^2).\]
\end{lemma}
\begin{proof}
	There are $\ceil{N/r}$ bucket vertices and $O(N/r+r)$ edge-bucket vertices by Lemma~\ref{lem:bucket-labels}.
	The copied test board has $O(r^2)$ vertices.
\end{proof}

\paragraph{Markers and lists.}
We now introduce markers, which are auxiliary names used only to define list constraints.
Each source and target vertex receives a~set of markers.
A~source vertex $v$ may be mapped precisely to those target vertices $x$ for~which $A_{\Lambda}(v)\subseteq A_{\Gamma}(x)$, where $A_{\Lambda}$ and $A_{\Gamma}$ denote the marker sets associated with the vertices of~$\Lambda$ and~$\Gamma$, respectively.
Throughout this construction, we write $\beta_L\coloneqq\left\lceil\log_2 L\right\rceil$ and $\beta_r\coloneqq\left\lceil\log_2 r\right\rceil$ for the label- and coordinate-index lengths.
The markers are:
\begin{itemize}
	    \item selector markers $\rho_U,\rho_W,\rho_A,\rho_B,\rho_X,\rho_Y,\rho_C$;
	    \item for every label bit $h\in[\beta_L]$ and bit value $b\in\{0,1\}$, markers
	    $\mathsf D_{h,b}$, $\mathsf A_{h,b}$, and $\mathsf B_{h,b}$;
	    \item for every coordinate bit $h\in[\beta_r]$ and bit value $b\in\{0,1\}$,
	    markers $\mathsf X_{h,b}$ and $\mathsf Y_{h,b}$;
	    \item for every edge-index bit $h\in[\beta_r]$ and bit value $b\in\{0,1\}$, a~marker
	    $\mathsf P_{h,b}$;
	    \item color markers $\mathsf C_1,\mathsf C_2,\mathsf C_3$.
\end{itemize}
Thus the number of markers is
\[
        M \coloneqq 7+6\beta_L+ 6\beta_r+3=O(\log r).
\]
For $m\in\{L,r\}$ and $z\in[m]$, let $\operatorname{bin}_m(z)$ be the fixed-length binary encoding of $z-1\in\{0,\ldots,m-1\}$ with $\left\lceil\log_2 m\right\rceil$ bits.
For a label $\ell\in[L]$, define
\[
\begin{aligned}
        \mathsf D(\ell)&\coloneqq\{\mathsf D_{h,\operatorname{bin}_L(\ell)_h}\colon h\in[\beta_L]\},\\
        \mathsf A(\ell)&\coloneqq\{\mathsf A_{h,\operatorname{bin}_L(\ell)_h}\colon h\in[\beta_L]\},\\
        \mathsf B(\ell)&\coloneqq\{\mathsf B_{h,\operatorname{bin}_L(\ell)_h}\colon h\in[\beta_L]\}.
\end{aligned}
\]
For a~coordinate $x\in[r]$, define
\[
\begin{aligned}
        \mathsf X(x)&\coloneqq\{\mathsf X_{h,\operatorname{bin}_r(x)_h}\colon h\in[\beta_r]\},\\
        \mathsf Y(x)&\coloneqq\{\mathsf Y_{h,\operatorname{bin}_r(x)_h}\colon h\in[\beta_r]\}.
\end{aligned}
\]
For an~edge index $p\in[r]$, define
\[
        \mathsf{Pcode}(p)\coloneqq\{\mathsf P_{h,\operatorname{bin}_r(p)_h}\colon h\in[\beta_r]\}.
\]

\paragraph{Target marker sets.}
The target marker sets are
\[
        A_\Gamma(U_{\ell,c})=\{\rho_U\}\cup \mathsf D(\ell),
\]
\[
        A_\Gamma(W_{Z,c,\Pi})=\{\rho_W\},
\]
and, for test vertices,
\[
        A_\Gamma(A_{p,a})=\{\rho_A\}\cup\mathsf{Pcode}(p)\cup\mathsf A(a),\qquad
        A_\Gamma(B_{p,b})=\{\rho_B\}\cup\mathsf{Pcode}(p)\cup\mathsf B(b),
\]
\[
        A_\Gamma(X_{p,x})=\{\rho_X\}\cup\mathsf{Pcode}(p)\cup\mathsf X(x),\qquad
        A_\Gamma(Y_{p,y})=\{\rho_Y\}\cup\mathsf{Pcode}(p)\cup\mathsf Y(y),
\]
\[
        A_\Gamma(C_{p,\gamma})=\{\rho_C,\mathsf C_\gamma\}\cup\mathsf{Pcode}(p).
\]

\paragraph{Source marker sets.}
The source marker sets are
\begin{align*}
		A_\Lambda(u_i)&=\{\rho_U\}\cup \mathsf D(\lambda(B_i)), \\
		A_\Lambda(w_j)&=\{\rho_W\},
\end{align*}
and each copied test vertex has the marker set of its target namesake, for example, $A_\Lambda(A'_{p,a})=\{\rho_A\}\cup\mathsf{Pcode}(p)\cup\mathsf A(a)$, and the same applies to $B'_{p,b}$, $X'_{p,x}$, $Y'_{p,y}$, and $C'_{p,\gamma}$.

The lists are defined by marker containment:
\[
        \mathcal L_r(v)=\{x\in V(\Gamma_r)\colon A_\Lambda(v)\subseteq A_\Gamma(x)\}.
\]
Thus each bucket vertex $u_i$ lists precisely to bucket-states of label $\lambda(B_i)$, each copied test vertex has a~singleton list, and each edge-bucket vertex lists to all edge-states.

\paragraph{Correctness.}

For completeness, a~proper $3$-coloring of~$Q$ is read off into the target: each bucket vertex goes to the bucket-state holding its true color string, and each edge-bucket vertex goes to the edge-state holding those strings and the map for its edges.

\begin{lemma}
If $Q$ is $3$-colorable, then $(\Lambda_r,\Gamma_r,\mathcal L_r)$ has a~list homomorphism.
\end{lemma}

\begin{proof}
	Let $\kappa\colon V(Q)\to[3]$ be a proper coloring. Map the copied test board $T'$ to the identically named tests in $T$.

	For a bucket $B_i$, let $c_i\in[3]^r$ be any string satisfying $c_i[v]=\kappa(v)$ for every vertex $v\in B_i$; fill unused coordinates arbitrarily.
	Set
	\[
			\phi(u_i)=U_{\lambda(B_i),c_i}.
	\]

	For an edge-bucket $E_j$, let
	\[I_j=\{i\colon B_i\text{ participates in }E_j\}\]
	and define
	\[Z_j=\{\lambda(B_i)\colon i\in I_j\}.\]
	The labels in $Z_j$ are distinct by Lemma~\ref{lem:bucket-labels}. Define strings
	$d_\ell$ for $\ell\in Z_j$ by $d_{\lambda(B_i)}=c_i$ for $i\in I_j$. This is
	well-defined because the labels $\lambda(B_i)$, $i\in I_j$, are pairwise distinct.

	Define a partial function $\Pi_j$ as follows.
	If the edge with index $p$ in $E_j$ is oriented from coordinate $x$ in bucket $B_\alpha$ to coordinate $y$ in bucket $B_\beta$, then set
	\[
			\Pi_j(p)=(\lambda(B_\alpha),x,\lambda(B_\beta),y).
	\]
	Let $d=(d_\ell)_{\ell\in Z_j}$.
	Set
	\[
			\phi(w_j)=W_{Z_j,d,\Pi_j}.
	\]

	The assigned images lie in the required lists by the marker definitions.
	If $w_ju_i$ is an edge of $\Lambda_r$, then $B_i$ participates in $E_j$, so the chosen edge-state explicitly contains the bucket-state assigned to $u_i$; hence the state edge is present in $\Gamma_r$.

	Finally consider the test edges from $w_j$ for the edge with index $p$, oriented from coordinate $x$ in bucket $B_\alpha$ to coordinate $y$ in bucket $B_\beta$.
	Put $a=\lambda(B_\alpha)$ and $b=\lambda(B_\beta)$.
	By construction, $\Pi_j(p)=(a,x,b,y)$, so the metadata edges to $A_{p,a}$, $X_{p,x}$, $B_{p,b}$, and $Y_{p,y}$ are present in $\Gamma_r$.
	Since $\kappa$ is proper, the two endpoint colors are not both $\gamma$ for any $\gamma\in[3]$, and therefore all color-check edges to $C_{p,\gamma}$ are present.
\end{proof}

For soundness, the lists force every source vertex into its intended type, the state edges make neighboring buckets agree on their shared color strings, and a~missing color-check edge would witness a~monochromatic edge of~$Q$, which a~list homomorphism cannot produce.

\begin{lemma}
If $(\Lambda_r,\Gamma_r,\mathcal L_r)$ has a~list homomorphism, then $Q$ is $3$-colorable.
\end{lemma}

\begin{proof}
Let $\phi\colon \Lambda_r\to \Gamma_r$ be a list homomorphism.

Because $\phi(u_i)\in\mathcal L_r(u_i)$, the marker definitions imply that $\phi(u_i)$ is a bucket-state with label $\lambda(B_i)$.
Thus
\[
        \phi(u_i)=U_{\lambda(B_i),c_i}
\]
for some $c_i\in[3]^r$.

Similarly, $\phi(w_j)\in\mathcal L_r(w_j)$ implies
\[
        \phi(w_j)=W_{Z_j,d,\Pi_j}
\]
for some edge-state.
Since $w_ju_i$ is an edge whenever $B_i$ participates in $E_j$, the definition of state edges forces
\[
        \lambda(B_i)\in Z_j
        \quad\text{and}\quad
        d_{\lambda(B_i)}=c_i
        \qquad\text{for every participating }B_i.
\]

Each copied test vertex has a singleton list determined by its selector and binary code.
For example, $\mathcal L_r(A'_{p,a})=\{A_{p,a}\}$: the selector chooses the $A$-family, and the two binary codes determine the edge index and label.
Thus $A'_{p,a}$ maps to $A_{p,a}$, and the same argument applies to the $B$, $X$, $Y$, and $C$ tests.

The strings $c_i$ define a coloring of all real vertices of $Q$ by assigning each $v\in B_i$ the color $c_i[v]$.
Suppose some edge is monochromatic.
Let this edge lie in $E_j$, and suppose, with the orientation used in the construction, that its endpoints are position $x$ in $B_\alpha$ and position $y$ in $B_\beta$.
Put
\[\gamma=c_\alpha[x]=c_\beta[y].\]
Let $p$ be the index of this edge.
The graph $\Lambda_r$ contains the metadata edges from $w_j$ to $A'_{p,\lambda(B_\alpha)}$, $X'_{p,x}$, $B'_{p,\lambda(B_\beta)}$, and $Y'_{p,y}$, and it also contains the color-check edge $w_jC'_{p,\gamma}$.
The image of $w_j$ is an edge-state whose strings on the labels $\lambda(B_\alpha)$ and $\lambda(B_\beta)$ are exactly $c_\alpha$ and $c_\beta$, by the state edges to $u_\alpha$ and $u_\beta$.
The metadata edges force $\Pi_j(p)=(\lambda(B_\alpha),x,\lambda(B_\beta),y)$.
Hence the condition $c_\alpha[x]=\gamma=c_\beta[y]$ causes the corresponding color-check edge to $C_{p,\gamma}$ to be absent in $\Gamma_r$.
This contradicts that $\phi$ is a homomorphism.
Therefore the induced coloring of $Q$ is proper.
\end{proof}

\begin{corollary}
For every input graph $Q$ of maximum degree at most $4$,
\[
        Q \text{ is } 3\text{-colorable}
        \quad\Longleftrightarrow\quad
        (\Lambda_r,\Gamma_r,\mathcal L_r)\text{ has a list homomorphism}.
\]
\end{corollary}

\subsection{Removing the Lists}
\label{subsec:remove-lists}

A~\emph{frame} is a~small rigid graph equipped with many pairwise disjoint groups of vertices, the \emph{ports}.
Its purpose is to simulate, inside ordinary graph homomorphism, the list constraints that the reduction needs: we attach a~source vertex to the ports that encode its intended type, and rigidity guarantees that these incidences survive every homomorphism.
Concretely, if a~source vertex is complete to a~chosen set of ports (that is, adjacent to every vertex of each of these ports), then the image of that vertex under any homomorphism must be complete to the corresponding target ports.

We build the frame in~two steps.
First we construct a~$3$-uniform hypergraph that admits no nontrivial automorphism yet contains $M$ disjoint perfect matchings into triples.

\begin{lemma}\label{lem:rigid-triple-system}
	For every $M\ge1$, there is a~$3$-uniform hypergraph $\mathcal A=(C,\mathcal E)$ and pairwise edge-disjoint subfamilies $\mathcal C_1,\ldots,\mathcal C_M\subseteq\mathcal E$ such that $|C|=q=O(M)$, $q\ge6$, $3\mid q$, every $\mathcal C_i$ is a~partition of $C$ into triples, $|\mathcal E|=O(M^2)$, and $\operatorname{Aut}(\mathcal A)$ is trivial.
\end{lemma}

\begin{proof}
	Let $N=\max\{M,10\}$ and put
	\[
			C=C_A\sqcup C_B\sqcup C_C, \text{ where } C_A=\{a_1,\ldots,a_N\},\quad C_B=\{b_1,\ldots,b_N\},\quad C_C=\{c_1,\ldots,c_N\}.
	\]
	Thus
	\[
			q=|C|=3N=O(M),\qquad q\ge 30,\qquad 3\mid q.
	\]

	For $i\in[M]$, set $r_i=i-1$ and define
	\[
			\mathcal C_i=\bigl\{\{a_x,b_{x+r_i},c_{x+2r_i}\}\colon x\in[N]\bigr\},
	\] where subscripts are read modulo $N$, with representatives in $[N]$.

	Each $\mathcal C_i$ is a~partition of $C$ into triples.
	The families $\mathcal C_1,\ldots,\mathcal C_M$ are pairwise edge-disjoint.
	Indeed, if
	\[
			\{a_x,b_{x+r_i},c_{x+2r_i}\}=\{a_y,b_{y+r_j},c_{y+2r_j}\},
	\]
	then the disjointness of $C_A,C_B,C_C$ gives $x=y$ and $r_i\equiv r_j\pmod N$.
	Since $0\le r_i,r_j\le M-1\le N-1$, we get $r_i=r_j$, and hence $i=j$.

	It remains to add marker edges that destroy all automorphisms.
	Let $s$ be the least integer such that $\binom{s}{2}\ge q$.
	Since $N\ge 10$, we have $\binom N2\ge 3N=q$, so $s\le N$.
	Set
	\[
			S=\{a_1,\ldots,a_s\}\subseteq C_A,\qquad \ell=|C\setminus S|=q-s.
	\]
	List the pairs of $S$ in lexicographic order as $P_1,P_2,\ldots,P_{\binom{s}{2}}$.
	Enumerate $C\setminus S$ as $\{v_1,\ldots,v_\ell\}$.
	For each $j\in[\ell]$, add the marker edges (since $\ell\le q\le\binom{s}{2}$, the first $\ell$ pairs are defined)
	\[
			P_t\cup\{v_j\},\qquad 1\le t\le j.
	\]
	Let $\mathcal R$ be the resulting marker family, and define
	\[
			\mathcal E=\mathcal R\sqcup \mathcal C_1\sqcup\cdots\sqcup\mathcal C_M.
	\]

	We now prove that marker degrees distinguish all vertices.
	For a~vertex $v_j\in C\setminus S$, we have $\deg_{\mathcal R}(v_j)=j$, so their degrees are precisely $[\ell]$.
	For others, write
	\[
			d_i=\deg_{\mathcal R}(a_i),\qquad 1\le i\le s.
	\]
	If $P_t$ contains $a_i$, then this pair contributes $\ell-t+1$ to $d_i$, because it appears with precisely the vertices $v_j$ satisfying $j\ge t$.

	Since $P_1, \ldots, P_{\binom{ s }{ 2 }}$ are in lexicographic order, $d_1>d_2>\cdots>d_s$.
	We also claim that every $d_i$ exceeds $\ell$.
	Since $q\ge 30$, the minimality of $s$ implies $s\ge 9$.
	If $s=9$, then $q\ge 30=4s-6$.
	If $s\ge 10$, then
	\[
			q\ge \binom{s-1}{2}+1\ge 4s-6.
	\]
	Thus $\ell=q-s\ge 3s-6$.
	The pairs $\{a_1,a_s\},\ \{a_2,a_s\},\ \{a_3,a_s\}$ occur in the lexicographic list at positions $s-1$, $2s-3$, and $3s-6$, respectively.
	All three are therefore used, and they all contain~$a_s$.
	Consequently
	\[
			d_s \ge \bigl(\ell-(s-1)+1\bigr)+\bigl(\ell-(2s-3)+1\bigr)+\bigl(\ell-(3s-6)+1\bigr)=3\ell-6s+13\ge \ell+1.
	\]
	Hence $d_1>d_2>\cdots>d_s>\ell$.
	All vertices of $C$ therefore have distinct marker degrees.

	Every vertex belongs to exactly one edge of each $\mathcal C_i$, so every vertex has total degree $M+\deg_{\mathcal R}(v)$ in $\mathcal A=(C,\mathcal E)$.
	The total degrees are distinct, and every automorphism preserves total degree.
	Thus every automorphism fixes every vertex, and $\operatorname{Aut}(\mathcal A)=\{1\}$.

	Finally,
	\[
			|\mathcal E|=MN+|\mathcal R|=MN+\sum_{j=1}^{\ell}j\le MN+\frac{q(q+1)}2=O(M^2).
	\]
\end{proof}

Now we turn this hypergraph skeleton into a~graph by blowing up each hyperedge into a~rigid graph gadget, which yields the frame together with its $M$ ports.
The gadget is a~copy of the fixed graph $J$ provided by the next lemma.

\begin{lemma}\label{lem:core}
There is a~connected, triangle-free graph $J$ with $\chi(J)=4$ that is a~core, that is, every endomorphism of $J$ is an~automorphism.
\end{lemma}

\begin{figure}[ht]
\centering
\begin{tikzpicture}[
    x=1cm,
    y=1cm,
    label distance=3pt,
    vert/.style={circle, draw=black!60, fill=black!8, minimum size=5pt, inner sep=0pt},
    centervertex/.style={circle, draw=orange!80!black, fill=orange!18, minimum size=6pt, inner sep=0pt},
    edge/.style={line width=0.5pt, draw=black!65, line cap=round, line join=round},
    colorone/.style={draw=blue!70!black, fill=blue!18},
    colortwo/.style={draw=teal!70!black, fill=teal!18},
    colorthree/.style={draw=violet!75!black, fill=violet!16},
    colorfour/.style={draw=orange!85!black, fill=orange!22},
    vertexlabel/.style={font=\scriptsize, inner sep=1pt}
]
    \foreach \i/\ang in {1/90,2/162,3/234,4/306,5/18} {
        \coordinate (u\i) at (\ang:2.2cm);
        \coordinate (v\i) at (\ang:1.1cm);
    }
    \coordinate (w) at (0,0);
    \foreach \i/\j in {1/2,2/3,3/4,4/5,5/1} {\draw[edge] (u\i) -- (u\j);}
    \draw[edge] (v1)--(u5); \draw[edge] (v1)--(u2);
    \draw[edge] (v2)--(u1); \draw[edge] (v2)--(u3);
    \draw[edge] (v3)--(u2); \draw[edge] (v3)--(u4);
    \draw[edge] (v4)--(u3); \draw[edge] (v4)--(u5);
    \draw[edge] (v5)--(u4); \draw[edge] (v5)--(u1);
    \foreach \i in {1,...,5} {\draw[edge] (w) -- (v\i);}
    \foreach \i/\ang/\col in {1/90/colorone,2/162/colortwo,3/234/colorone,4/306/colortwo,5/18/colorthree} {
        \node[vert, \col, label={[vertexlabel]\ang:$u_{\i}$}] at (u\i) {};
        \node[vert, \col, label={[vertexlabel]\ang:$v_{\i}$}] at (v\i) {};
    }
    \node[centervertex, colorfour, label={[vertexlabel]270:$w$}] at (w) {};
\end{tikzpicture}
\caption{The Gr{\"o}tzsch graph.
	It is connected, triangle-free, $4$-chromatic, and every endomorphism is an~automorphism.}
\label{fig:grotzsch}
\end{figure}
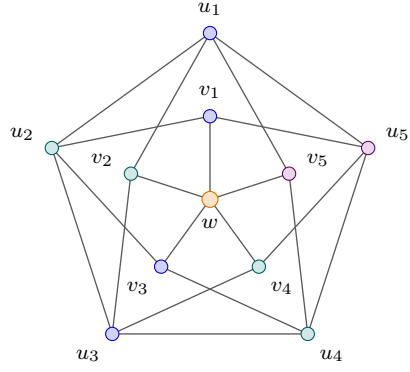
\begin{proof}
	Let~$J$ be the Gr{\"o}tzsch graph~\cite{Grotzsch59}, depicted in~\Cref{fig:grotzsch}.
	It~is connected, triangle-free, and satisfies $\chi(J)=4$.
	Moreover, every proper subgraph of~$J$ is $3$-colorable.
	If an~endomorphism were not onto the whole vertex set, its image would be a~proper subgraph and hence $3$-colorable, giving a~$3$-coloring of~$J$, a~contradiction.
	If it is onto the whole vertex set, then, since $J$ is finite, it is bijective and therefore an~automorphism.
\end{proof}

We turn the construction in~\Cref{lem:rigid-triple-system} into a~graph frame.
The vertex set $C$ of the hypergraph becomes a~clique, whose vertices we call \emph{central}; being the unique maximum clique, it must be preserved by every homomorphism, so a~homomorphism induces a~permutation $\sigma$ of~$C$.
Each hyperedge is then realized by a~graph gadget, which is a~copy of a~fixed template, attached so that it can survive only if $\sigma$ maps hyperedges to hyperedges.

\begin{lemma}\label{lem:frame}
	For every integer $M\ge 1$, there~is a~graph $F_M$ with pairwise disjoint non-empty port sets $P_1,\ldots,P_M\subseteq V(F_M)$ such that the following hold.
	\begin{itemize}
		\item $|V(F_M)|=O(M^2)$ and $|P_i|=O(M)$ for every $i$.
		\item No vertex of $F_M$ is adjacent to every vertex of any port $P_i$.
		\item Let $B$ be any bipartite graph. Form a graph $X$ from the disjoint union
		of $F_M$ and $B$ by making each $b\in V(B)$ complete to a union of ports
		$\bigcup_{i\in A(b)}P_i$, where $A(b)\subseteq [M]$. Then every homomorphism
		$F_M\to X$ is an automorphism of $F_M$ and maps every port $P_i$ onto itself.
	\end{itemize}
\end{lemma}
\begin{proof}
	Let $J$ be the graph from~\Cref{lem:core}, and denote $g=|V(J)|$.
	Apply Lemma~\ref{lem:rigid-triple-system} and let $\mathcal A=(C,\mathcal E)$ be the resulting $3$-uniform hypergraph, with $q=|C|$.
	Construct $F_M$ as follows. Make $C$ a clique. For every hyperedge
	$e\in\mathcal E$, add a fresh copy $J_e$ of $J$. There are no edges between
	distinct copies $J_e$ and $J_f$. For $x\in V(J_e)$ and $c\in C$, add the edge $xc$
	exactly when $c\notin e$. Thus every vertex of $J_e$ is adjacent to all central
	vertices except the three vertices of $e$.

	For $i\in[M]$, define
	\[
			P_i=\bigcup_{e\in\mathcal C_i}V(J_e).
	\]
	The ports are non-empty and pairwise disjoint because the families $\mathcal C_i$
	are pairwise edge-disjoint. Also
	\[
			|V(F_M)|=|C|+g|\mathcal E|=O(M^2)
	\]
	and, since each $\mathcal C_i$ partitions $C$ into triples,
	\[
			|P_i|=g|\mathcal C_i|=O(M).
	\]

	We next verify that no vertex of $F_M$ is complete to any port.
	Fix $i\in[M]$.
	For $c\in C$, the partition $\mathcal C_i$ contains some triple $e$ with
	$c\in e$, and every vertex of $J_e\subseteq P_i$ is non-adjacent to $c$.
	For $x\in V(J_f)$ for some $f\in\mathcal E$, since $q\ge 6$, $\mathcal C_i$ contains
	at least two triples. Choose $e\in\mathcal C_i$ with $e\ne f$. There are no edges
	between distinct gadget copies, so $x$ has no neighbors in $J_e\subseteq P_i$.

	It remains to prove the rigidity property. Let $B$ be bipartite, and form $X$ from
	the disjoint union of $F_M$ and $B$ by making every $b\in V(B)$ complete to the
	ports indexed by $A(b)\subseteq[M]$, with no other edges between $B$ and $F_M$.
	Let $\varphi\colon F_M\to X$ be a homomorphism.

	First, $C$ is the unique $q$-clique of $X$, since inside
	$F_M$, any clique not contained in $C$ uses vertices from at most one copy $J_e$.
	Because $J$ is triangle-free, it uses at most two vertices from $J_e$, and it can
	also use only central vertices from $C\setminus e$. Hence its size is at most
	\[
			2+|C\setminus e|=q-1.
	\]
	If a clique uses vertices of $B$, then it uses no central vertex, because no vertex
	of $B$ is adjacent to $C$. Moreover $B$ has clique number at most $2$, and the
	clique can use vertices from at most one copy $J_e$, again at most two of them.
	Thus such a clique has size at most $4<q$. Therefore $C$ is the unique $q$-clique
	of $X$.

	Since the image of the clique $C$ under $\varphi$ is a $q$-clique in the loopless
	graph $X$, we have $\varphi(C)=C$. Let $\sigma=\varphi|_C$, viewed as a permutation
	of $C$.

	Fix $e\in\mathcal E$. Every vertex of $J_e$ is adjacent to every central vertex in
	$C\setminus e$, so every image of a vertex of $J_e$ is adjacent to every vertex of
	\[
			\varphi(C\setminus e)=C\setminus\sigma(e).
	\]
	Now consider different possibilities for $\varphi$ on $J_e$.
	No vertex of $B$ qualifies, because vertices of $B$ have no neighbors in $C$.
	A~central vertex
	$c\in C$ qualifies exactly when $c\in\sigma(e)$, since the graph is loopless.
	A~gadget vertex in $J_f$ qualifies exactly when
	\[
			C\setminus\sigma(e)\subseteq C\setminus f,
	\]
	or equivalently $f\subseteq\sigma(e)$. Since both sets have size $3$, this means
	$f=\sigma(e)$.

	Consequently, $\varphi$ can map a~vertex of $J_e$ to either the part of central clique $K_3$ on $\sigma(e)$ or to $J_{\sigma(e)}$ (if hyperedge $\sigma(e)$ exists, i.e. $\sigma(e) \in \mathcal{E}$).
	Note that there are no edges between the central triple $\sigma(e)$	and $J_{\sigma(e)}$.

	The graph $J_e$ is connected, so its image lies in one connected component of this
	common-neighbor set. It cannot lie in the central $K_3$, because there is no
	homomorphism $J\to K_3$. Hence $\sigma(e)\in\mathcal E$ and
	\[
			\varphi(V(J_e))\subseteq V(J_{\sigma(e)}).
	\]

	Thus $\sigma(e)\in\mathcal E$ for every $e\in\mathcal E$. Since $\sigma$ is a
	permutation of $C$, this gives $\sigma(\mathcal E)=\mathcal E$, so $\sigma$ is an
	automorphism of the hypergraph $\mathcal A$, therefore
	$\sigma=\operatorname{id}_C$.
	It follows that $\varphi$ fixes $C$ pointwise and maps
	every copy $J_e$ isomorphically onto itself, since the map $J_e\to J_{\sigma(e)}$ is an~endomorphism and hence it is an~automorphism.

	Therefore $\varphi$ is an automorphism of $F_M$ whose image is contained in $F_M$.
	Finally, every port is a union of whole copies $J_e$,
	\[
			P_i=\bigcup_{e\in\mathcal C_i}V(J_e),
	\]
	and each such copy is mapped onto itself. Hence $\varphi(P_i)=P_i$ for every
	$i\in[M]$.
\end{proof}

For a vertex $z$ outside the frame and a set $A\subseteq [M]$, we say that $z$ has
anchor set $A$, meaning that $z$ is adjacent to every vertex in every port $P_i$ with
$i\in A$, and to no other vertex of the frame.

\begin{corollary}\label{lem:anchor}
Let $X$ be as in Lemma~\ref{lem:frame}.
Let $X'$ contain a copy $F'_M$ of $F_M$, and
let $\phi\colon X'\to X$ be a homomorphism. Then $\phi$ maps every port of $F'_M$ onto the
corresponding port of $F_M$. Hence, if a vertex $z\in V(X')\setminus V(F'_M)$ is
complete to the ports indexed by $A$, then $\phi(z)$ is complete to the same ports
in $X$.
\end{corollary}

\begin{proof}
By Lemma~\ref{lem:frame}, the restriction of $\phi$ to $F'_M$ is a port-preserving
isomorphism onto the target copy of $F_M$. Therefore every port is mapped onto its
corresponding port. If $z$ is complete to the ports indexed by $A$, then for every
$i\in A$ and every $p\in P_i$, the edge $zp$ of $X'$ maps to an edge between $\phi(z)$
and $\phi(p)$. Since $\phi(P_i)=P_i$, the vertex $\phi(z)$ is complete to $P_i$ for
every $i\in A$.
\end{proof}

\paragraph{Compiling the lists.}

We now convert the marked list-homomorphism instance into an~ordinary homomorphism instance.
The needed property for the reduction is that every list is defined by containment of a~non-empty marker set over the same $M=O(\log r)$ markers.

\begin{lemma}\label{lem:frame-compilation}
There are graphs $G_r$ and $H_r$ such that
\[
        (\Lambda_r,\Gamma_r,\mathcal L_r)\text{ has a list homomorphism}
        \quad\Longleftrightarrow\quad
        G_r\to H_r.
\]
Moreover, $|V(G_r)|=|V(\Lambda_r)|+O((\log r)^2)$ and $|V(H_r)|=|V(\Gamma_r)|+O((\log r)^2)$.
\end{lemma}

\begin{proof}
	Take the frame $F_M$ from \Cref{lem:frame}, and identify its ports with the $M$ markers.
	Construct $H_r$ from the disjoint union of $F_M$ and $\Gamma_r$ by making every target vertex $x\in V(\Gamma_r)$ complete to precisely the ports in $A_\Gamma(x)$.
	Construct $G_r$ from the disjoint union of a~copy $F'_M$ of the frame and $\Lambda_r$ by making every source vertex $v\in V(\Lambda_r)$ complete to precisely the ports in $A_\Lambda(v)$.
	No other edges are added.
	The graph $\Gamma_r$ is bipartite with sides $S_E$ and $S_B\cup T$: all state and testing edges run between these two sides, and there are no edges inside $S_B$, inside $S_E$, inside $T$, or between $S_B$ and $T$.
	Thus the target has the form required by \Cref{lem:frame,lem:anchor}.

	If $\psi\colon\Lambda_r\to \Gamma_r$ is a~list homomorphism, extend it by any fixed isomorphism $F'_M\to F_M$.
	The edges of $\Lambda_r$ are preserved because $\psi$ is a~homomorphism.
	For every source vertex $v\in V(\Lambda_r)$, the list condition gives $A_\Lambda(v)\subseteq A_\Gamma(\psi(v))$, so all frame edges incident with $v$ are also preserved.
	Thus the extension is a~homomorphism $G_r\to H_r$.

	Conversely, let $\phi\colon G_r\to H_r$ be a~homomorphism.
	By \Cref{lem:anchor}, the copied frame maps portwise to the target frame.
	Every marker set $A_\Lambda(v)$ is non-empty, and no frame vertex is complete to any port, so $\phi(v)$ cannot be a~frame vertex for $v\in V(\Lambda_r)$.
	Hence $\phi(v)\in V(\Gamma_r)$.
	Because $v$ is complete to every port in $A_\Lambda(v)$ and the frame ports are preserved, $\phi(v)$ is complete to every corresponding target port.
	By construction of $H_r$, this means $A_\Lambda(v)\subseteq A_\Gamma(\phi(v))$, so $\phi(v)\in\mathcal L_r(v)$.
	Finally, every edge of $\Lambda_r$ maps to an edge of $\Gamma_r$, since the only non-frame edges of $H_r$ inside $V(\Gamma_r)$ are the edges of $\Gamma_r$.
	Thus $\phi|_{V(\Lambda_r)}$ is a~list homomorphism.

	The size bounds follow from $|V(F_M)|=O(M^2)=O((\log r)^2)$.
\end{proof}

\begin{corollary}
For every input graph $Q$ of maximum degree at most $4$,
\[
        Q \text{ is } 3\text{-colorable}
        \quad\Longleftrightarrow\quad
        G_r\to H_r .
\]
\end{corollary}

\begin{lemma}
\[|V(G_r)|=O(N/r+r^2).\]
\end{lemma}

\begin{proof}
This follows from \Cref{lem:list-source-size} and $|V(F_M)|=O((\log r)^2)=O(r^2)$.
\end{proof}

\begin{lemma}
\[
        |V(H_r)|\le 2^{O(r^2)}.
\]
\end{lemma}

\begin{proof}
By \Cref{lem:target-size} and \Cref{lem:frame-compilation},
\[
        |V(H_r)|=|V(\Gamma_r)|+O((\log r)^2)\le 2^{O(r^2)}.
\]
\end{proof}

\subsection{Finishing the Lower Bound}
\label{subsec:parameters-eth}

\paragraph{Target degeneracy.}
Recall that degeneracy is the least~$d$ for which the vertices can be ordered so that each has at most~$d$ neighbors later in~the order.
We eliminate the edge-states first: although there may be $2^{O(r^2)}$ of them, each sees only the bucket-states of the labels it stores and the tests for its own edge indices, so it has degree $O(r)$ (\Cref{lem:edge-state-degree}).
Once the edge-states are gone, every remaining non-frame vertex sees only its $O((\log r)^2)$ frame anchors (\Cref{lem:residual}), and the frame itself has only $O((\log r)^2)$ vertices.
The whole order therefore witnesses degeneracy $O(r)$.
This upper bound is the only degeneracy estimate used in~the lower-bound proof.

\begin{lemma}\label{lem:edge-state-degree}
Every edge-state has degree $O(r)$ in $H_r$.
\end{lemma}

\begin{proof}
Let $W=W_{Z,c,\Pi}\in S_E$.
Its neighbors in $S_B$ are exactly the bucket-states $U_{\ell,c_\ell}$ with $\ell\in Z$.
Hence
\[
        |N(W)\cap S_B|\le |Z|\le 2r.
\]
For every edge index in $\operatorname{dom}(\Pi)$, the edge-state has at most four metadata neighbors and at most three color-check neighbors.
Therefore
\[
        |N(W)\cap T|\le 7|\operatorname{dom}(\Pi)|\le 7r.
\]
Its frame neighbors lie only in the port $\rho_W$.
Since that port has size $O(\log r)$,
\[
        |N(W)\cap V(F_M)|=O(\log r).
\]
Thus $\deg_{H_r}(W)=O(r)$.
\end{proof}

\begin{lemma}\label{lem:residual}
In $H_r-S_E$, every vertex of $S_B\cup T$ has degree $O((\log r)^2)$.
\end{lemma}

\begin{proof}
After deleting $S_E$, there are no state edges and no testing edges left.
Every bucket-state or test vertex is adjacent only to $O(\log r)$ frame ports, and every port has size $O(\log r)$.
\end{proof}

\begin{theorem}\label{thm:target-deg}
\[
        \operatorname{degen}(H_r)=O(r).
\]
\end{theorem}

\begin{proof}
Use the following elimination order.
First remove all vertices of $S_E$ in an arbitrary order.
By Lemma~\ref{lem:edge-state-degree}, each such vertex has at most $O(r)$ surviving neighbors when it is removed.

Then remove all vertices of $S_B\cup T$.
By Lemma~\ref{lem:residual}, after the edge-states are deleted their residual degree is $O((\log r)^2)=O(r)$.

	Finally remove the frame. The frame has $O((\log r)^2)=O(r)$ vertices, so every
	remaining frame vertex has residual degree at most $O(r)$. Hence every vertex is
	removed with at most $O(r)$ later neighbors.
\end{proof}

\paragraph{ETH consequence.}

\mainDegLBFormal*
\begin{proof}[Proof of \Cref{thm:main-lb-formal}]
	Fix a~nondecreasing unbounded function $D\colon\mathbb N\to\mathbb N$ satisfying $D(n)=O(n^{1/3})$ and $D(2n)=O(D(n))$.
	Suppose that an~algorithm with the claimed running time exists on every instance $(G,H)$ satisfying $\operatorname{degen}(H)\le D(n)$.
	Let $Q$ be an~$N$-vertex graph of maximum degree at most~$4$, and choose $m$ minimally so~that $mD(m)\ge N$.
	Monotonicity and $D(2m)=O(D(m))$ imply $mD(m)=\Theta(N)$.
	Let $\kappa\ge1$ be a~constant from~\Cref{thm:target-deg} such that $\operatorname{degen}(H_r)\le\kappa r$, and set $r=\lfloor D(m)/\kappa\rfloor$.
	Since $D$ is unbounded and $D(m)=O(m^{1/3})$, we have $r\to\infty$ and $r^3=O(m)=o(N)$, so the construction yields
	\[|V(G_r)|=O(N/D(m)+D(m)^2)=O(m),\qquad |V(H_r)|=2^{O(D(m)^2)}=2^{o(N)}.\]
	Choose a~constant integer $c$ such that $|V(G_r)|\le cm$ for all sufficiently large~$N$, and pad $G_r$ with isolated vertices to a~graph $\widehat G_r$ on $n=cm$ vertices.
	Since $H_r$ has an~edge, this padding preserves whether a~homomorphism exists, while monotonicity and repeated use of the doubling condition give
	\[\operatorname{degen}(H_r)\le\kappa r\le D(m)\le D(n),\qquad D(n)n=O(D(m)m)=O(N).\]
	Thus the assumed algorithm, together with the construction, solves $Q$ in~time $2^{o(N)}$, contradicting \Cref{thm:bdcolor}.
\end{proof}

\section{The Degeneracy Threshold from One to Two}
\label{sec:constant-degenerate-nlogn}

This section proves the super-exponential lower bound for $\problemabbr{HOM}(\cdot,\mathcal{H})$ even when $\mathcal{H}$ has constant degeneracy.
The reduction uses the~following bounded-alphabet $2$-\problemabbr{CSP} lower bound of~\cite{Marx10,KMPS24}.

\begin{theorem}[{\cite[Theorem 1.6]{KMPS24}}]
\label{thm:sparse-bounded-alphabet-csp}
	Let $f$ be a~nondecreasing, unbounded, and polynomial-time computable function.
	Unless $\cc{ETH}$ fails, there is no algorithm running in~time
	\[f(k-1)^{o(k/\log k)}\]
	that decides all $2$-\problemabbr{CSP} instances whose constraint graph is a~simple bipartite $3$-regular graph on~$k$ vertices and whose alphabet size is smaller than $f(k)$.
\end{theorem}

\begin{theorem}[Formal version of \Cref{thm:constant-deg-lb}]
	\label{thm:constant-deg-lb-formal}
	Let $f$ be a~nondecreasing, unbounded, and polynomial-time computable function.
	There is a~polynomial-time reduction\footnote{Here, polynomial time means polynomial in~the~combined size of~the~input and~the~output.} that maps every $2$-\problemabbr{CSP} instance covered by \Cref{thm:sparse-bounded-alphabet-csp}, whose constraint graph has~$k$ vertices, to~an~equivalent instance $(G,H)$ of~$\problemabbr{HOM}$ with $N=|V(G)|=\Theta(k)$ and
	\[\operatorname{degen}(H)\le2,\qquad |V(H)|+|E(H)|=O(kf(k)^2),\]
	for which no~algorithm with running time
	\[O^*\left(f(k-1)^{o(k/\log k)}\right)\]
	exists, unless $\cc{ETH}$ fails.

\end{theorem}

At a~high level, the proof first reduces a~sparse $\problemabbr{CSP}$ instance $\Gamma$ from~\Cref{thm:sparse-bounded-alphabet-csp} to $\problemabbr{HOM}(G_{\Gamma}, H_\Gamma)$.
Then, we take a~subdivision of both graphs (i.e., we replace every edge of both graphs by a~path of length three).
This fixed subdivision preserves the existence of a~homomorphism and makes the target $2$-degenerate.

However, the price of this step is an~increase in size.
After subdivision, the source has $|V(G_\Gamma)|+O(|E(G_\Gamma)|)$ vertices, and the target has $|V(H_\Gamma)|+O(|E(H_\Gamma)|)$ vertices.
Thus the subdivision can be used in~a fine-grained lower bound only when the source graph produced by the reduction is sparse.

\subsection{Sparse Pinning Frame}

We begin with the~construction of~a~frame from~\cite{FGKM15} and then state its~properties.
For an~integer $s\ge1$, define a~graph $F_s$ as follows.
It has $s + 1$ distinguished vertices $z_0,z_1,\ldots,z_s$.
For every $i\in[s]$, add four more vertices $a_i,b_i,c_i,d_i$.
Add the~edges of the~$5$-cycle
\[z_{i-1}a_i,\quad a_ib_i,\quad b_ic_i,\quad c_id_i,\quad d_iz_{i-1},\]
and add the~five edges from $z_i$ to this cycle,
\[z_iz_{i-1},\quad z_ia_i,\quad z_ib_i,\quad z_ic_i,\quad z_id_i.\]
Let $Z_s=\{z_1,\ldots,z_s\}$ be the~set of these cycle-completing vertices.
\Cref{fig:htree} gives an~example.

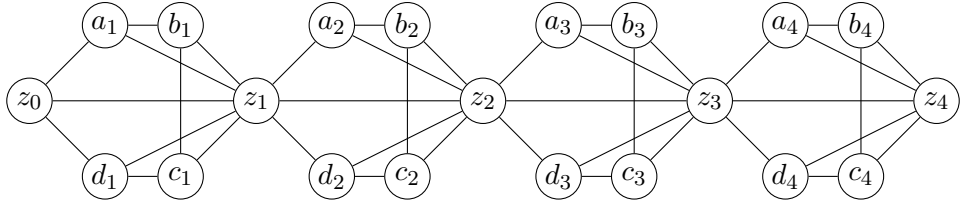
\begin{figure}
    \begin{center}
        \begin{tikzpicture}
            \tikzstyle{v} = [circle, draw, inner sep=0mm, minimum size=6mm]

            \foreach \a/\i/\x/\y in {z/0/0/0,
                a/1/1/1, b/1/2/1, z/1/3/0, c/1/2/-1, d/1/1/-1,
                a/2/4/1, b/2/5/1, z/2/6/0, c/2/5/-1, d/2/4/-1,
                a/3/7/1, b/3/8/1, z/3/9/0, c/3/8/-1, d/3/7/-1,
                a/4/10/1, b/4/11/1, z/4/12/0, c/4/11/-1, d/4/10/-1
            }
                \node[v] (\a\i) at (\x, \y) {$\a_\i$};

            \foreach \i in {1,...,4} {
                \tikzmath{\j=int(\i-1);}
                \draw (z\i) -- (z\j);
                \draw (z\i) -- (a\i);
                \draw (z\i) -- (b\i);
                \draw (z\i) -- (c\i);
                \draw (z\i) -- (d\i);
                \draw (z\j) -- (a\i) -- (b\i) -- (c\i) -- (d\i) -- (z\j);
            }
        \end{tikzpicture}
    \end{center}
    \caption{The graph~$F_4$.}
    \label{fig:htree}
\end{figure}

\begin{lemma}
\label{lem:four-frame}
	The graph $F_s$ has degeneracy at most $3$.
	Moreover, the~following statements hold.
	\begin{enumerate}
		\item If $x\notin Z_s$, then $F_s[N_{F_s}(x)]$ is bipartite.
		\item If $x=z_i\in Z_s$, then $F_s[N_{F_s}(z_i)]$ contains the~$5$-cycle $C_i=z_{i-1}a_ib_ic_id_iz_{i-1}$ and is non-bipartite.
		\item For every $i\in[s]$, the~unique non-bipartite connected component of $F_s[N_{F_s}(z_i)]$ is $C_i$.
		If $i\ge2$, the~only vertex of~$Z_s$ in~this component is $z_{i-1}$; if $i=1$, this component contains no vertex of~$Z_s$.
	\end{enumerate}
\end{lemma}
\begin{proof}
	Every vertex among $a_i,b_i,c_i,d_i$ has degree exactly $3$ in~$F_s$.
	Hence every subgraph containing one of these vertices has a~vertex of degree at most $3$.
	If a~subgraph contains no such vertex, then it is a~subgraph of the~path on~$z_0,z_1,\ldots,z_s$, and hence has a~vertex of degree at most $1$ unless it is empty.
	Therefore $F_s$ is $3$-degenerate.

	The neighborhood of every vertex outside $Z_s$ is bipartite by direct inspection.
	For example, $N(z_0)=\{a_1,d_1,z_1\}$ induces the~path $a_1-z_1-d_1$, and $N(a_i)=\{z_{i-1},b_i,z_i\}$ induces a~path with center $z_i$.
	The cases $b_i,c_i,d_i$ are the~same up to symmetry along the~cycle.

	Inside $N(z_i)$, the~vertices of $C_i$ induce exactly that $5$-cycle.
	The other frame neighbors of $z_i$, when they exist, lie in~the three-vertex path $a_{i+1}-z_{i+1}-d_{i+1}$, which is disjoint from $C_i$ inside the~open neighborhood.
	Thus $C_i$ is the~unique non-bipartite connected component.
	For $i\ge2$, the~only vertex of~$Z_s$ among the~vertices of~$C_i$ is $z_{i-1}$; for $i=1$, the~cycle $C_i$ contains no vertex of~$Z_s$.
\end{proof}

The~next step is to~use the~distinguished vertices in~$Z_s$ as pin locations for the~types that occur in~the \problemabbr{CSP} reduction.
In that reduction, \problemabbr{CSP} variables and \problemabbr{CSP} edges both serve as types: for a~\problemabbr{CSP} variable $v$ and a~possible value $a$, the~vertex $A_{v,a}$ has type $v$.
For a~\problemabbr{CSP} edge $e=(u,v)$ and an~allowed pair $(a,b)$, where $a$ corresponds to $u$ and $b$ corresponds to $v$, the~vertex $R_{e,a,b}$ has type $e$.
The vertex $R_{e,a,b}$ is adjacent to the~pin $p_e$ and to the~value vertices $A_{u,a}$ and $A_{v,b}$, whose types are $u$ and~$v$, respectively; both are different from the~constraint type $e$.

Let $T$ be a~finite set of types, let $L=|T|$, put $s=10L+10$ and construct $F_s$, then enumerate $T=\{t_1,\ldots,t_L\}$.
For $t_j$, define two pins
\[p_{t_j}=z_{10j},\qquad q_{t_j}=z_{10j+5}.\]
The spacing ensures that all pins are distinct and that the~two pins of any type have no common frame neighbor.
For an~edge type $e$, only the~pin $p_e$ will be used by constraint vertices.

We say that a~graph is obtained from $F_s$ by safe attachments if new vertices are added only in~the following two forms.
A~value vertex of type $t$ is a~new vertex adjacent to the~two pins $p_t$ and $q_t$.
A~constraint vertex of type $t$ is a~new vertex adjacent to the~pin $p_t$ and to two value vertices whose types are different from~$t$.
See~\Cref{fig:safe-attachments}.

\begin{figure}
\centering
\begin{tikzpicture}[
    x=1cm,
    y=1cm,
    line cap=round,
    line join=round,
    pinvertex/.style={circle, draw=black!65, fill=white, minimum size=6.2pt, inner sep=0pt},
    valuevertex/.style={circle, draw=black!65, fill=white, minimum size=7.2pt, inner sep=0pt},
    constraintvertex/.style={circle, draw=black!65, fill=white, minimum size=7.8pt, inner sep=0pt},
    pinedge/.style={draw=black!62, line width=0.55pt},
    constraintedge/.style={draw=black!62, line width=0.6pt},
    vertexlabel/.style={font=\scriptsize, inner sep=1pt, text=black!78},
    typelabel/.style={font=\scriptsize, inner sep=1pt, text=black!62}
]
    \path[use as bounding box] (-3.65,-1.95) rectangle (3.65,1.72);

    \coordinate (pu) at (-2.8,0.9);
    \coordinate (qu) at (-2.8,-0.9);
    \coordinate (pe) at (0,0.9);
    \coordinate (qe) at (0,-0.9);
    \coordinate (pv) at (2.8,0.9);
    \coordinate (qv) at (2.8,-0.9);

    \coordinate (Au) at (-1.45,0);
    \coordinate (Re) at (0,0);
    \coordinate (Av) at (1.45,0);

    \draw[pinedge] (Au) -- (pu);
    \draw[pinedge] (Au) -- (qu);
    \draw[pinedge] (Av) -- (pv);
    \draw[pinedge] (Av) -- (qv);
    \draw[constraintedge] (Re) -- (pe);
    \draw[constraintedge] (Re) -- (Au);
    \draw[constraintedge] (Re) -- (Av);

    \node[pinvertex, label={[vertexlabel]above:$p_u$}] at (pu) {};
    \node[pinvertex, label={[vertexlabel]below:$q_u$}] at (qu) {};
    \node[pinvertex, label={[vertexlabel]above:$p_v$}] at (pv) {};
    \node[pinvertex, label={[vertexlabel]below:$q_v$}] at (qv) {};
    \node[pinvertex, label={[vertexlabel]above:$p_e$}] at (pe) {};
    \node[pinvertex, label={[vertexlabel]below:$q_e$}] at (qe) {};

    \node[valuevertex, label={[vertexlabel]above:$A_{u,a}$}] at (Au) {};
    \node[valuevertex, label={[vertexlabel]above:$A_{v,b}$}] at (Av) {};
    \node[constraintvertex, label={[vertexlabel, xshift=2pt]above right:$R_{e,a,b}$}] at (Re) {};

    \node[typelabel] at (-1.45,-1.5) {type $u$};
    \node[typelabel] at (0,-1.5) {type $e$};
    \node[typelabel] at (1.45,-1.5) {type $v$};
\end{tikzpicture}

\caption{A local view of safe attachments.}
\label{fig:safe-attachments}
\end{figure}
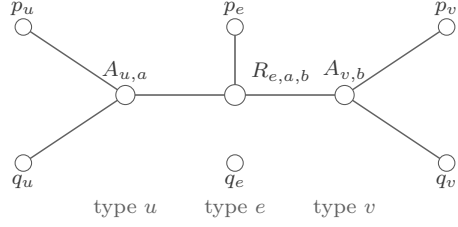

\begin{lemma}
\label{lem:four-pinning}
	Let $H$ be any graph obtained from $F_s$ by safe attachments.
	Then every homomorphism $\psi\colon F_s\to H$ fixes every vertex of~$Z_s$:
	\[\psi(z_i)=z_i,\qquad\text{for all }i\in[s].\]
\end{lemma}
\begin{proof}
	First, the~only vertices of $H$ whose open neighborhoods are non-bipartite are $z_1,\ldots,z_s$.
	Indeed, frame vertices outside $Z_s$ keep the~same open neighborhoods as in~$F_s$, because attachments are made only to vertices of~$Z_s$.
	A~value vertex has a~bipartite open neighborhood: apart from constraint vertices, it sees only its two nonadjacent pins, and the~constraint vertices adjacent to it are not adjacent to those pins.
	A~constraint vertex has an~independent open neighborhood consisting of one pin and two value vertices.

	Moreover, for $z_i\in Z_s$, the~added vertices in~$N_H(z_i)$ are isolated in~the induced graph $H[N_H(z_i)]$.
	If an~added value vertex is adjacent to $z_i$, then its other pin is not adjacent to $z_i$ by the~spacing of pins.
	If an~added constraint vertex is adjacent to $z_i$, then its other neighbors are two value vertices, none of which is adjacent to $z_i$.
	Consequently, the~unique non-bipartite component of $H[N_H(z_i)]$ is still the~cycle component $C_i$ from Lemma~\ref{lem:four-frame}.

	Let $\psi\colon F_s\to H$ be a~homomorphism.
	For each $i\in[s]$, the~cycle $C_i$ is contained in~$N_{F_s}(z_i)$.
	Its image is an~odd closed walk in~$H[N_H(\psi(z_i))]$.
	Hence $H[N_H(\psi(z_i))]$ is non-bipartite, and therefore $\psi(z_i)\in Z_s$.
	Let $\psi(z_i)=z_{r_i}$ with $r_i\in[s]$.

	For $i\ge2$, the~cycle $C_i$ contains the~vertex $z_{i-1}\in Z_s$.
	Since $C_i$ is connected and contains an~odd cycle, its image lies in~the unique non-bipartite component $C_{r_i}$ of $H[N_H(z_{r_i})]$.
	The vertex $z_{i-1}$ maps to a~vertex of~$Z_s$, and the~only vertex of~$Z_s$ in~$C_{r_i}$ is $z_{r_i-1}$.
	Thus $r_i\ge2$ and $\psi(z_{i-1})=z_{r_i-1}$.
	Equivalently, $r_{i-1}=r_i-1$ for every $i=2,\ldots,s$.
	Since $1\le r_s\le s$, we get $r_1\le1$.
	Also $r_1\ge1$, so $r_1=1$ and $r_i=i$ for every $i\in[s]$.
\end{proof}

\subsection{Reduction}
\label{subsec:reduction}

Let
\[\Gamma=(P,\Sigma,\{C_e\}_{e\in E(P)})\]
be a~$2$-\problemabbr{CSP} instance whose constraint graph $P$ is a~simple bipartite $3$-regular graph on~$k$ vertices.
Let $q=|\Sigma|$.
Since $P$ is $3$-regular, $|E(P)|=3k/2$.
Fix an~arbitrary orientation of every edge $e\in E(P)$.

Define the~type set
\[T=V(P)\cup E(P).\]
Then $|T|=k+3k/2=5k/2$.
Build the~pinning frame $F_s$ with $s=10|T|+10$ and with the~pin pairs $(p_t,q_t)$ for $t\in T$ as above.

The source graph $G_\Gamma$ starts with a~copy of $F_s$.
For every \problemabbr{CSP} variable $v\in V(P)$, add a~vertex $x_v$ and the~two pin edges $x_vp_v$ and $x_vq_v$.
For every oriented constraint edge $e=(u,v)\in E(P)$, add a~vertex $y_e$ and the~three edges
\[y_ep_e,\qquad y_ex_u,\qquad y_ex_v.\]

The target graph $H_\Gamma$ starts with a~copy of the~same frame $F_s$.
For every \problemabbr{CSP} variable $v\in V(P)$ and every value $a\in\Sigma$, add a~value vertex $A_{v,a}$ and the~two pin edges $A_{v,a}p_v$ and $A_{v,a}q_v$.
For every oriented constraint edge $e=(u,v)$ and every allowed pair $(a,b)\in C_e$, add a~constraint vertex $R_{e,a,b}$ and the~three edges
\[R_{e,a,b}p_e,\qquad R_{e,a,b}A_{u,a},\qquad R_{e,a,b}A_{v,b}.\]

\begin{lemma}
\label{lem:four-correctness}
	The \problemabbr{CSP} instance $\Gamma$ is satisfiable if and only if $G_\Gamma\to H_\Gamma$.
\end{lemma}
\begin{proof}
	Suppose first that $\Gamma$ has a~satisfying assignment $\alpha\colon V(P)\to\Sigma$.
	Map the~frame in~$G_\Gamma$ identically to the~frame in~$H_\Gamma$.
	For every variable vertex $x_v$, set $\phi(x_v)=A_{v,\alpha(v)}$.
	For every constraint edge $e=(u,v)$, the~pair $(\alpha(u),\alpha(v))$ belongs to $C_e$, so the~constraint vertex $R_{e,\alpha(u),\alpha(v)}$ exists; set $\phi(y_e)=R_{e,\alpha(u),\alpha(v)}$.
	All frame edges, pin edges, and incidence edges are preserved by construction, so $\phi$ is a~homomorphism.

	Conversely, let $\phi\colon G_\Gamma\to H_\Gamma$ be a~homomorphism.
	The restriction of $\phi$ to the~frame copy in~$G_\Gamma$ fixes every vertex of~$Z_s$ by Lemma~\ref{lem:four-pinning}.
	In particular, all pins are fixed.

	For a~\problemabbr{CSP} variable $v$, the~vertex $x_v$ is adjacent to both $p_v$ and $q_v$.
	Therefore $\phi(x_v)$ is a~common neighbor of $p_v$ and $q_v$ in~$H_\Gamma$.
	The pin pair has no common neighbor inside the~frame, and the~only added common neighbors are the~value vertices $A_{v,a}$ with $a\in\Sigma$.
	Thus there is a~unique value $\alpha(v)\in\Sigma$ such that $\phi(x_v)=A_{v,\alpha(v)}$.

	Now let $e=(u,v)$ be a~constraint edge.
	Since $y_e$ is adjacent to $p_e$, $x_u$, and $x_v$, the~vertex $\phi(y_e)$ is a~common neighbor of $p_e$, $A_{u,\alpha(u)}$, and $A_{v,\alpha(v)}$.
	No frame vertex has this property: a~frame neighbor of $A_{u,\alpha(u)}$ must be $p_u$ or $q_u$, while a~frame neighbor of $A_{v,\alpha(v)}$ must be $p_v$ or $q_v$, and these four pins are distinct by construction because $u\ne v$.
	No value vertex has this property either, since value vertices have no value neighbors.
	Hence $\phi(y_e)$ is a~constraint vertex.
	A~constraint vertex adjacent to $p_e$ has type $e$, and its adjacency to $A_{u,\alpha(u)}$ and $A_{v,\alpha(v)}$ forces this vertex to be $R_{e,\alpha(u),\alpha(v)}$.
	Therefore $(\alpha(u),\alpha(v))\in C_e$.
	This holds for every constraint edge, so $\alpha$ satisfies $\Gamma$.
\end{proof}

\begin{lemma}
\label{lem:four-size}
	The reduction takes $\Gamma = (P, \Sigma, \{ C_e \}_{e \in E(P)})$, constructs $G_\Gamma$ and $H_\Gamma$ in~time polynomial in~the output size and satisfies
	\[|V(G_\Gamma)|=\Theta(k),\qquad |E(G_\Gamma)|=O(k)\]
	and
	\[|V(H_\Gamma)|+|E(H_\Gamma)|=O(kq^2),\]
	where $k = |V(P)|$, and $q = |\Sigma|$.
\end{lemma}
\begin{proof}
	The frame has $|V(F_s)|=(s+1)+4s=5s+1$ vertices, where $s=10|T|+10=25k+10$, so $|V(F_s)|=\Theta(k)$.
	The source graph adds one vertex for every vertex of $P$ and one vertex for every edge of $P$, hence adds $k+3k/2=O(k)$ vertices.
	Therefore $|V(G_\Gamma)|=\Theta(k)$.
	The source graph has the~$O(k)$ frame edges, two pin edges for every \problemabbr{CSP} variable, and three edges for every \problemabbr{CSP} edge vertex.
	Thus $|E(G_\Gamma)|=O(k)$.

	The target graph has the~frame, $kq$ value vertices, and at most
	\[\sum_{e\in E(P)} |C_e|\le |E(P)|q^2=O(kq^2)\]
	constraint vertices.
	Each added value vertex contributes two pin edges, and each constraint vertex contributes three edges.
	The frame has $O(k)$ edges.
	Therefore $|V(H_\Gamma)|+|E(H_\Gamma)|=O(kq^2)$, and the~construction time is polynomial in~this output size.
\end{proof}

For a~graph $X$, let $S_3(X)$ denote the graph obtained from $X$ by replacing every edge by a~path of length~$3$.

\begin{lemma}
\label{lem:three-subdivision-hom}
	For all finite loopless graphs $G$ and $H$,
	\[G\to H\qquad\Longleftrightarrow\qquad S_3(G)\to S_3(H).\]
\end{lemma}
\begin{proof}
	The forward implication is immediate.
	A~homomorphism $\varphi\colon G\to H$ maps each edge $uv\in E(G)$ to an~edge $\varphi(u)\varphi(v)\in E(H)$, and the subdivided path replacing $uv$ can be mapped to the subdivided path replacing $\varphi(u)\varphi(v)$.

	For the reverse implication, choose an arbitrary orientation of each edge of $H$.
	If an oriented edge $ab$ is replaced in $S_3(H)$ by the path $a-x_{ab}-y_{ab}-b$, define a~map $\rho\colon V(S_3(H))\to V(H)$ by
	\[\rho(a)=a\text{ for original vertices }a,\qquad \rho(x_{ab})=b,\qquad \rho(y_{ab})=a.\]
	This map sends every edge of $S_3(H)$ to an~edge of $H$.

	We claim that if $w_0w_1w_2w_3$ is a~walk of length~$3$ in $S_3(H)$, then $\rho(w_0)\rho(w_3)\in E(H)$.
	Indeed, every two-step walk centered at an internal subdivision vertex has endpoints with the same $\rho$-image.
	If $w_1$ is internal, then $\rho(w_0)=\rho(w_2)$, and the edge $w_2w_3$ gives $\rho(w_2)\rho(w_3)\in E(H)$.
	If $w_1$ is original, then $w_2$ is internal, so $\rho(w_1)=\rho(w_3)$, and the edge $w_0w_1$ gives $\rho(w_0)\rho(w_1)\in E(H)$.
	This proves the claim.

	Now let $\psi\colon S_3(G)\to S_3(H)$ be a~homomorphism.
	For every original vertex $v\in V(G)$, define $\varphi(v)=\rho(\psi(v))$.
	If $uv\in E(G)$, then the subdivided path replacing $uv$ in $S_3(G)$ maps under $\psi$ to a~walk of length~$3$ in $S_3(H)$ from $\psi(u)$ to $\psi(v)$.
	Hence, $\varphi(u)\varphi(v)\in E(H)$ and $\varphi \colon G\to H$ is a~homomorphism.
\end{proof}

\begin{lemma}
\label{lem:three-subdivision-parameters}
	For every graph $H$, the graph $S_3(H)$ satisfies
	\[\operatorname{degen}(S_3(H))\le2.\]
\end{lemma}
\begin{proof}
	Let $J$ be a~non-empty subgraph of $S_3(H)$.
	If $J$ contains an internal subdivision vertex, then that vertex has degree at most~$2$ in $S_3(H)$, and hence degree at most~$2$ in $J$.
	If $J$ contains no internal subdivision vertex, then $J$ is an~independent set of original vertices.
	Thus every non-empty subgraph of $S_3(H)$ has a~vertex of degree at most~$2$.
\end{proof}

\begin{proof}[Proof of Theorem~\ref{thm:constant-deg-lb-formal}]
	Apply Theorem~\ref{thm:sparse-bounded-alphabet-csp} with the function $f$ from the statement, and let $\Gamma$ be an~instance with constraint graph on~$k$ vertices and alphabet size $q<f(k)$.

	Construct $(G_\Gamma,H_\Gamma)$ as in~\Cref{subsec:reduction}, in~time polynomial in~the output size by~\Cref{lem:four-size}, and then set
	\[\widehat G_\Gamma=S_3(G_\Gamma),\qquad \widehat H_\Gamma=S_3(H_\Gamma).\]
	By Lemma~\ref{lem:four-correctness} and Lemma~\ref{lem:three-subdivision-hom}, $\Gamma$ is satisfiable if and only if $\widehat G_\Gamma\to \widehat H_\Gamma$.
	By Lemma~\ref{lem:three-subdivision-parameters}, the target satisfies $\operatorname{degen}(\widehat H_\Gamma)\le2$.
	Let $N=|V(\widehat G_\Gamma)|$.
	By Lemma~\ref{lem:four-size},
	\[N=|V(G_\Gamma)|+2|E(G_\Gamma)|=\Theta(k).\]
	The same size bound gives
	\[|V(\widehat H_\Gamma)|+|E(\widehat H_\Gamma)|=O(|V(H_\Gamma)|+|E(H_\Gamma)|)=O(kq^2)=O(kf(k)^2).\]
	The construction time is polynomial in~this output size.
	Therefore, an~algorithm solving all constructed instances in~time
	\[f(k-1)^{o(k/\log k)}\cdot (|V(G)|+|V(H)|)^{O(1)}\]
	would give an~$f(k-1)^{o(k/\log k)}$-time algorithm for the \problemabbr{CSP} instances covered by Theorem~\ref{thm:sparse-bounded-alphabet-csp}, contradicting that theorem.
\end{proof}

\section{No-Compression Barriers for SAT Reductions}
\label{sec:no-compression}

\Cref{thm:main-lb} shows that degeneracy is visible linearly in the exponent of $\problemabbr{HOM}$, but the hard targets behind it are large: the construction needs a~target of size $h=2^{\Theta(D^2)}$ to realize degeneracy scale~$D$.
This is polynomial in~$n$ only while $D=O(\sqrt{\log n})$; for larger degeneracy the target becomes superpolynomial.
It is therefore natural to ask whether the same linear dependence can be forced while the target stays polynomially bounded.

This section gives a~conditional reason to expect that it cannot, at least not through the standard route.
Fine-grained lower bounds for $\problemabbr{HOM}$, like those for many neighboring problems, are ultimately reductions from sparse $3$-$\problemabbr{SAT}$ via $\cc{ETH}$ and the Sparsification Lemma.
We prove that a~polynomial target reduction would ``compress'' the $\problemabbr{SAT}$ instance, thereby contradicting \Cref{hyp:nocomp}.

The argument has four parts.
First, we discuss why an~explicit sparse $3$-$\problemabbr{SAT}$ instance carries hardness only at scale $s/\log s$, where $s$~is its bit length, and we recall the No-Compression Hypothesis stating no ordinary reduction does better.
Second, we apply this argument to sparse $\problemabbr{CSP}$ instances with a~linear number of constraints, showing
that one cannot improve the known lower bounds for~it under~\Cref{hyp:nocomp}.
Third, we show that degenerate targets compress the boundary information of homomorphism into short list records; this lets~us reduce $\problemabbr{HOM}$ to another problem which is compressible.
Finally, we show the resulting no-compression barrier for $\problemabbr{HOM}$.

\subsection{Sparse-SAT Information Scale}

The Sparsification Lemma~\cite{IPZ01} lets one run $\cc{ETH}$ through sparse $3$-$\problemabbr{SAT}$: formulas on $N$ variables with $O(N)$ clauses still require $2^{\Omega(N)}$ time under $\cc{ETH}$.
If the problem is measured by an~explicit bit length~$s$, then to encode a~sparse formula one needs $s = \Theta(N\log N)$.
Hence the hardness of sparse $3$-$\problemabbr{SAT}$ is \[2^{\Omega(N)}=2^{\Omega(s/\log s)}.\] under $\cc{ETH}$.

To prove a~lower bound \emph{above} this scale for an~$s$-bit target problem, a~sparse-$\problemabbr{SAT}$ reduction must do some compression of the formula.
A~broad variety of fine-grained lower bounds is known, yet to our knowledge not one of them establishes a~bound better than $2^{\Omega(s/\log s)}$ in the input bit length~$s$.

This distinction has a~precedent in the barrier of Kulikov and Mihajlin for the \problemname{Edge Coloring} problem~\cite{KM24}.
The input to~this problem has size $n^2$ and
the best known upper bound is~of the~form $2^{O(n^2)}$.
Proving an~$\alpha^{n^2}$ lower bound for this problem under $\cc{SETH}$ for any $\alpha > 0$ would break $\cc{SETH}$.
The hypothesis we~use is~the $\cc{ETH}$ analogue of that barrier, and it does not follow from their argument.

We now state the formal version of the no-compression principle used in~the rest of~the section.
\begingroup
\renewcommand{\thetheorem}{\ref{hyp:nocomp}}
\renewcommand{\theHtheorem}{hyp.nocomp.formal}
\begin{hypothesis}[No-Compression Hypothesis]\label{hyp:nocomp-formal}
	Fix an~explicit encoding of a~decision problem $\mathcal A$, and let $s$ be the bit length of an~encoded instance.
	There is no fine-grained reduction from sparse $3$-$\problemabbr{SAT}$ on $N$ variables to $\mathcal A$, and no function $\Psi(s)=\omega(s/\log s)$, with the following property: if $\mathcal A$ has an~algorithm running in time $2^{o(\Psi(s))}$ on $s$-bit instances, then sparse $3$-$\problemabbr{SAT}$ has an~algorithm running in time $2^{o(N)}$.

	Equivalently, there is no reduction establishing under $\cc{ETH}$ a~lower bound \[2^{\Omega(\Psi(s))}\] for any $\Psi(s)=\omega(s/\log s)$.
\end{hypothesis}
\addtocounter{theorem}{-1}
\endgroup

\subsection{Application to Sparse CSP}
\label{subsec:sparse-csp-nocomp}

\begin{theorem}
	Assume \Cref{hyp:nocomp-formal}.
	Let $k=k(n)$ and $h=h(n)$.
	Set $B(n)=k\log n+h^k$.
	Suppose that
	\[\log h=\omega\!\left(\frac{B(n)}{\log n + \log B(n)}\right).\]
	Then no~fine-grained reduction from sparse $3$-$\problemabbr{SAT}$ can establish, under $\cc{ETH}$, an~$h^{\Omega(n)}$-time lower bound for sparse $(k,h)$-$\problemabbr{CSP}$.
\end{theorem}

\begin{proof}
	Consider a~sparse $(k,h)$-$\problemabbr{CSP}$ instance with $n$ variables.
	Each constraint has arity at most~$k$, and there are at most~$Cn$ constraints for some $C > 0$.
	Thus the constraint scopes can be encoded using $O_C(nk\log n)$ bits.
	Each constraint relation has arity at most~$k$ and is written as a~truth table of size at most~$h^k$.
	Since the number of constraints is at most~$Cn$, all truth tables use $O_C(nh^k)$ bits.
	Hence the total bit length is
	\[s=O_C(nB(n)).\]

	Now suppose that a~fine-grained reduction established an~$h^{\Omega(n)}$ lower bound for this $\problemabbr{CSP}$ family under $\cc{ETH}$.
	By the assumption, $n\log h=\omega(nB(n)/\log(nB(n)))$.
	Since $s=O_C(nB(n))$, this implies $n\log h=\omega(s/\log s)$, contradicting \Cref{hyp:nocomp-formal}.
\end{proof}

Note that for fixed~$k$, the condition holds for every $h=h(n)\to\infty$ with $h^k=O(\log n)$.
However, for constant~$h$, one has $n\log h=\Theta(n)$ and $s/\log s=\Theta(n)$, so it does not rule out $\cc{ETH}$ lower bounds.

\subsection{Succinct List-Extension Homomorphism Problem}

Suppose a~partial homomorphism has already mapped some neighbors of an~unmapped source vertex.
The vertices of~$H$ still available to that vertex are exactly the common neighborhood in~$H$ of the images of its mapped neighbors.
The lemma says that in a~degenerate target this common neighborhood is either small or is cut out by few constraints.

\begin{lemma}\label{lem:small-list}
	Let $D\in\mathbb Z_{\ge 0}$, and let $H$ be a loopless graph with $\operatorname{degen}(H)\le D$. For every set $C\subseteq V(H)$, define
	\[
			A(C)=\bigcap_{c\in C}N_H(c),
	\]
	with the convention $A(\varnothing)=V(H)$. Then either
	\[
			|A(C)|\le D\qquad\text{or}\qquad|C|\le D.
	\]
\end{lemma}
\begin{proof}
	Suppose not. Then $|A(C)|\ge D+1$ and $|C|\ge D+1$.
	Since $H$ is loopless,
	$A(C)\cap C=\varnothing$: if $c\in C$, then $c\notin N_H(c)$, so $c\notin A(C)$.

	Choose subsets $A'\subseteq A(C)$ and $C'\subseteq C$ with $|A'|=|C'|=D+1$. By
	definition of $A(C)$, every vertex of $A'$ is adjacent to every vertex of $C'$. Thus
	$H$ contains $K_{D+1,D+1}$ as a subgraph. This graph has minimum degree $D+1$, and
	hence degeneracy would be at least $D+1$.
\end{proof}

Succinct list extension is the bookkeeping problem for exactly the information that survives this compression.
Each source vertex carries either the short list of colors still available to it or a~short certificate of that list, and the task is to complete the partial homomorphism subject to these lists.

\begin{definition}
	Fix functions $D=D(n)$ and $T=T(n)$.
	An instance of $\problemabbr{SLEH}_{D,T}$ consists of:
	\begin{enumerate}
		\item a source graph $G$ on at most $n$ vertices and at most $Tn$ edges;
		\item a~target graph $H$ with $h\le n$ and $\operatorname{degen}(H)\le D$;
		\item for every $x\in V(G)$, a succinct list record of one of the following two
		forms:
		\begin{enumerate}
			\item an explicit list $S_x\subseteq V(H)$ with $|S_x|\le D$, interpreted as
			$L(x)=S_x$;
			\item an intersection certificate $C_x\subseteq V(H)$ with $|C_x|\le D$,
			interpreted as
			\[
				L(x)=\bigcap_{c\in C_x}N_H(c).
			\]
		\end{enumerate}
	\end{enumerate}
	The task is to decide whether there exists a homomorphism $\psi\colon G\to H$ such that
	$\psi(x)\in L(x)$ for every $x\in V(G)$.
\end{definition}

Since $h\le n$, a~target vertex can be named using $O(\log n)$ bits, so each list record costs $O(D\log n)$ bits.
The source edge set costs $O(Tn\log n)$ bits, and the target, being $D$-degenerate, has at most $Dh$ edges and so is stored in $O(Dh\log h)=O(Dn\log n)$ bits.
Hence every such instance has bit size
\[s=O((T+D)n\log n).\]

\subsection{Reduction}

\begin{theorem}\label{thm:scheduled-core}
	Let $D=D(n)\ge 1$ and $\tau=\tau(n)\ge 1$.
	There is a~fine-grained Turing reduction from $\problemabbr{HOM}$ instances $(G,H)$ with source size~$n$ and target size~$h$ satisfying
	\[h\le n,\qquad \operatorname{degen}(H)\le D\]
	to $\problemabbr{SLEH}_{D,\tau}$ such that every oracle query has bit size
	\[s=O(n(D+\tau)\log n),\]
	and the total number of oracle queries, as well as the total non-oracle work, is at most
	\[2^{O\left(n\log D+\frac n\tau(\log h+\log \tau)\right)}.\]
\end{theorem}
\begin{proof}
	The idea is a~branching scheme that keeps the residual source graph sparse enough to hand to the oracle.
	We process the target in a~degeneracy order.
	At each step we either certify that the unmapped part of~$G$ has no dense core, in which case it is sparse and we ship it to the list-extension oracle, or we locate a~dense core and branch on its vertices.

	Fix an~ordering $x_1,\ldots,x_n$ of $V(G)$, and compute a~degeneracy ordering $w_1,\ldots,w_h$ of~$H$ such that every vertex has at most~$D$ neighbors to its right.
	For each target index $i \in [h]$ let
	\[R_i \coloneqq N_H(w_i)\cap\{w_{i+1},\ldots,w_h\},\qquad |R_i|\le D.\]

	Set $p_{\max}=\floor{\frac{n}{\tau+1}}$.
	The reduction first iterates over all sets $P\subseteq V(G)$ with $|P|\le p_{\max}$ and maps $\sigma\colon P\to[h]$.
	Now, fix some $P$ and $\sigma$.

	We then branch as follows.
	A~state of the branch is a~pair $(\alpha, i)$, where $M=\operatorname{dom}(\alpha)$ is~the set of currently mapped source vertices, $\alpha\colon M\to V(H)$ and $i\in[h+1]$.
	At every state we reject if $\alpha$ is not a partial homomorphism: if there is an~edge $uv\in E(G)$ with $u,v\in M$ and $\alpha(u)\alpha(v)\notin E(H)$, the branch is discarded.

	Let
	\[
			U=V(G)\setminus M
	\]
	be the current unmapped set, and let
	\[
			K=\operatorname{core}_\tau(G[U])
	\]
	be the $\tau$-core of the residual graph $G[U]$, i.e., the maximal induced subgraph
	of $G[U]$ of minimum degree at least $\tau$.

	If $K=\varnothing$, then $G[U]$ is $(\tau-1)$-degenerate. Therefore $|E(G[U])|\le \tau n$.
	We now output a~succinct list-extension oracle query. Its source graph is $G[U]$.
	For every $y\in U$, define the set of boundary vertices
	\[
			C_\alpha(y)=\{\alpha(x)\colon x\in M\cap N_G(y)\}.
	\]
	The vertices available to $y$ because of its already mapped neighbors are exactly
	\[
			A_\alpha(y)=\bigcap_{c\in C_\alpha(y)}N_H(c).
	\]
	If $|A_\alpha(y)|\le D$, give $y$ the explicit list
	record $A_\alpha(y)$. If $|A_\alpha(y)|>D$, then Lemma~\ref{lem:small-list} implies
	$|C_\alpha(y)|\le D$, and we give $y$ the intersection-certificate record
	$C_\alpha(y)$. This query is equivalent to asking whether $\alpha$ extends to a full
	homomorphism $G\to H$.

	Suppose now that $K\ne\varnothing$.
	If $i=h+1$, reject the branch.
	If there is no vertex $v\in P\cap V(K)$ with $\sigma(v)=i$, advance to state $(\alpha, i+1)$ and finish this branch.

	If such vertices exist, choose the first one and call it~$v$.
	Since $v\in K$, choose the first $\tau$ neighbors of~$v$ in~$K$ and call this set $S(v)$.
	Branch over all maps
	\[\beta\colon S(v)\to R_i.\]
	For each branch, set
	\[\alpha'=\alpha\cup\{v\mapsto w_i\}\cup\beta.\]
	If $\alpha'$ is inconsistent, discard it; otherwise recurse from $(\alpha',i)$.
	The index remains $i$ because other core vertices may also map to $w_i$.

	Soundness of the reduction is immediate from the construction of the
	list-extension query: if an~oracle query accepts, the list homomorphism on $G[U]$
	respects all residual edges and all boundary constraints to $M$, and therefore
	combines with $\alpha$ to give a~homomorphism $G\to H$.

	For completeness, suppose $\varphi\colon G\to H$ is a~homomorphism.
	We construct $P, \sigma$ and a~branch that finds some homomorphism $G \to H$.
	We follow a branch
	maintaining the invariant:
	\[
			\alpha(x)=\varphi(x)\text{ for every }x\in\operatorname{dom}(\alpha),
	\]
	and every vertex of the current $\tau$-core $K$ is mapped by $\varphi$ into the active
	suffix
	\[
			\{w_i,w_{i+1},\ldots,w_h\}.
	\]
	The invariant holds at the root $M = \varnothing, \alpha = \varnothing$.

	If $K=\varnothing$, the final oracle query is satisfied by $\varphi|_U$.
	If $K\ne\varnothing$ and no vertex of~$K$ maps to~$w_i$, increase~$i$.
	Otherwise choose the first vertex $v\in V(K)$ with $\varphi(v)=w_i$, add $v$ to~$P$, and set $\sigma(v)=i$.
	Let $S(v)$ be the first $\tau$ neighbors of~$v$ in~$K$.
	For every $u\in S(v)$, the edge $uv$ gives $\varphi(u)\in N_H(w_i)$; by the suffix invariant and looplessness, in fact $\varphi(u)\in R_i$.
	Thus the branch assigning $u\mapsto\varphi(u)$ for every $u\in S(v)$ is present.
	Extend $\alpha$ according to~$\varphi$ on $\{v\}\cup S(v)$ and continue with the same index~$i$.
	Each pivot step maps $\tau+1$ fresh vertices, so $P$ has at most $p_{\max}$ vertices at the end.
	Hence some branch reaches an~accepting oracle query for some $(P, \sigma)$.

	It remains to bound the oracle input bit size and the running time.
	At a~leaf, $G[U]$ has fewer than $\tau n$ edges, whose encoding costs $O(\tau n\log n)$ bits.
	The list records cost $O(Dn\log n)$ bits, and the $D$-degenerate target costs $O(Dh\log h) = O(Dn\log n)$ bits.
	Thus every query has bit size $O(n(D+\tau)\log n)$.

	The number of pairs $(P, \sigma)$ is at most
	\[\sum_{p\le p_{\max}}\binom np h^p\le 2^{O\left(n+\frac n\tau(\log h+\log \tau)\right)}.\]
	For a~fixed pair $(P, \sigma)$, the branches contribute at most $D^n=2^{O(n\log D)}$, because every source vertex is assigned at most once.
	This proves the stated bound.
\end{proof}

\subsection{No-Compression Barrier}

\begin{theorem}[Formal version of~\Cref{thm:unified-nocomp}]\label{thm:unified-nocomp-formal}
	Assume \Cref{hyp:nocomp-formal}.
	Let $h=h(n)\le n$ and $1\le D(n)\le n$, and let $\Lambda(n)=n\cdot\lambda(n)$ for some $\lambda(n)=\omega(\max\{D(n),\sqrt{\log h}\})$.
	Then there is no~fine-grained reduction from sparse $3$-$\problemabbr{SAT}$ on~$N$ variables to~$\problemabbr{HOM}$ restricted to~instances $(G,H)$ with source size~$n$, $|V(H)|\le h$, and $\operatorname{degen}(H)\le D(n)$, with the~following property: if~these instances admit an~algorithm running in~time $2^{o(\Lambda(n))}$, then sparse $3$-$\problemabbr{SAT}$ admits an~algorithm running in~time $2^{o(N)}$.
\end{theorem}

\begin{proof}
	Suppose toward contradiction that such a~fine-grained reduction establishes the lower bound $\Lambda(n) = n \cdot \omega(\max \{ D(n), \sqrt{ \log h } \})$.
	Set $\tau=\ceil {\max \{ D(n), \sqrt{ \log h } \}}$, and apply Theorem~\ref{thm:scheduled-core}.
	Hence every oracle query has bit size $s \coloneqq O(n \tau \log n)$.

	The overhead exponent is
	\[O\left(n\log D+\frac n\tau(\log h+\log \tau)\right) = O(n \cdot \tau).\]
	Since $\tau\le n$, each query size satisfies $s=O(n\tau\log n)$.
	Thus $\Lambda(n)=n\lambda(n)=\omega(n\tau)=\omega(s/\log s)$, while the reduction overhead is $2^{o(\Lambda(n))}$.
	Therefore an~algorithm for~the queried succinct list-extension instances with running time $2^{o(\Lambda(n))}$ would solve the promised $\problemabbr{HOM}$ instances in~time $2^{o(\Lambda(n))}$.
	Composing this with the assumed $3$-$\problemabbr{SAT}$ reduction would give a~$2^{o(N)}$ algorithm for $3$-$\problemabbr{SAT}$ on $N$ variables, contradicting $\cc{ETH}$.
	Equivalently, the composed proof would establish a~lower bound above the $s/\log s$ scale, contradicting \Cref{hyp:nocomp}.
\end{proof}

\section{Generated Formulas and Hard Witnesses}
\label{sec:generated-hard-witnesses}

This section shows that a~strong refutation of \Cref{hyp:nocomp} implies circuit lower bounds.
The~result of~Impagliazzo, Paturi, and~Zane~\cite{IPZ01} shows that $3$-$\problemabbr{SAT}$ remains hard even when the~number of~clauses is~linear.
Therefore, in~some sense, it provides a~generator that takes $O(n \log n)$ bits as~input and~outputs a~$3$-$\problemabbr{SAT}$ formula on~$n$ variables that is~hard under~$\cc{ETH}$.
Probably, the~most natural way to~refute \Cref{hyp:nocomp} is~to construct a~more efficient generator of~hard $3$-$\problemabbr{SAT}$ formulas.

We show that if a~short seed can generate hard formulas, then the generated formulas must sometimes have satisfying assignments of high circuit complexity.
Those hard assignments can then be exposed one bit at a~time by a~language in $\P^{\NP}$.
Any small circuit for this language, after hardwiring the generator seed, would give a~small circuit for the hard satisfying assignment, which is impossible.

The proof has three steps.
First we formally state what we mean by generators of hard formulas.
Second we show that this generator forces satisfying assignments of high circuit complexity: otherwise one could enumerate all small circuits, view each as an~assignment, and solve the generated formulas in $2^{o(N)}$ time.
Finally we package the lexicographically first satisfying assignment into a~language in $\P^{\NP}$, which is an~easy part.

\begin{definition}
	Fix $0<\alpha\le 1$ and set $r_\alpha(N)=\ceil{N^\alpha}$.
	An~\emph{$\alpha$-compressed hard formula generator} is a~uniform family of maps
	\[G_N\colon\{0,1\}^{r_\alpha(N)}\to\{\text{$3$-CNF formulas over $N$ variables}\}\]
	such that the following hold.
	\begin{itemize}
		\item $G_N(y)$ is computable in time $\operatorname{poly}(N)$.
		\item If there is an~algorithm $A$ and a~function $t(N)=2^{o(N)}$ such that, for every sufficiently large $N$ and every formula $F\in\operatorname{Im}(G_N)$, the algorithm $A$ on input $F$ decides whether it is satisfiable in time $t(N)$, then $\cc{ETH}$ is false.
	\end{itemize}
	Equivalently, assuming $\cc{ETH}$, $\problemabbr{SAT}$ restricted to $\operatorname{Im}(G_N)$ has no uniform $2^{o(N)}$-time algorithm.
\end{definition}

\begin{lemma}
\label{lem:generated-hard-witness}
	Assume $\cc{ETH}$, and let $G=\{G_N\}$ be an~$\alpha$-compressed hard formula generator.
	Let $s  \colon  \mathbb{N} \to \mathbb{N}$ be a~polynomial-time computable function satisfying
	\[s(N)\log(s(N)+\log N)=o(N).\]
	Then, for infinitely many $N$, there is a~seed $y_N\in\{0,1\}^{r_\alpha(N)}$ such that $G_N(y_N)$ is satisfiable, but every satisfying assignment $a\in\{0,1\}^N$ of $G_N(y_N)$ satisfies $\cc{CC}[a]>s(N)$.
\end{lemma}
\begin{proof}
	Suppose the conclusion fails.
	Then, for all sufficiently large $N$, every satisfiable formula in $\operatorname{Im}(G_N)$ has at least one satisfying assignment of circuit complexity at most $s(N)$.

	This gives a~$2^{o(N)}$-time algorithm for satisfiability on $\operatorname{Im}(G_N)$.
	On input $F\in\operatorname{Im}(G_N)$, enumerate all Boolean circuits on $\ceil{\log_2 N}$ inputs and of size at most~$s(N)$.
	The number of such circuits is at most
	\[2^{O(s(N)\log(s(N)+\log N))}=2^{o(N)}.\]
	For each circuit $C$, form the assignment $a_C=(C(0),C(1),\ldots,C(N-1))$ and test whether $F(a_C)=1$.
	Since $F$ has size $\operatorname{poly}(N)$, each test takes $\operatorname{poly}(N)$ time.

	If some assignment $a_C$ satisfies $F$, accept; otherwise reject.
	The algorithm is correct because a~satisfiable generated formula has a~small-circuit satisfying assignment by the contrary assumption, while an~unsatisfiable formula has no satisfying assignment at all.
	Its running time is
	\[2^{o(N)}\operatorname{poly}(N).\]
\end{proof}

Fix an~$\alpha$-compressed hard formula generator $G=\{G_N\}$.
We now turn satisfying assignments into a~language by asking for one indexed bit of the lexicographically first witness.
Use any fixed effective self-delimiting encoding of triples $\langle N,y,i\rangle$ such that, for fixed $N$ and~$y$, all choices of the index~$i \in \{ 0, 1 \}^{\ceil{\log_2 N}}$ have the same total input length.

\begin{definition}
	The language $\problemabbr{LexWit}_G$ is defined as follows.
	A~valid input is a~tuple $\langle N,y,i\rangle$ with $y\in\{0,1\}^{r_\alpha(N)}$ and $i\in\{0,1\}^{\ceil{\log_2 N}}$.
	Let $F=G_N(y)$.
	\begin{itemize}
		\item If $F$ is unsatisfiable, then $\problemabbr{LexWit}_G(N,y,i)=0$.
		\item If $F$ is satisfiable, let $a^{N,y}\in\{0,1\}^N$ be the lexicographically first satisfying assignment of~$F$, then
		\[\problemabbr{LexWit}_G(N,y,i)=\begin{cases} a^{N,y}_i & \text{if } i<N,\\ 0 & \text{if } i\ge N. \end{cases}\]
	\end{itemize}
\end{definition}

Let $\ell(N)$ denote the input bit size of $\problemabbr{LexWit}_G$, so $\ell(N) \coloneqq \Theta(N^{\alpha})$.

\begin{lemma}
	For every fixed $0<\alpha\le 1$,
	\[\problemabbr{LexWit}_G\in \P^{\NP}.\]
\end{lemma}
\begin{proof}
	On input $\langle N,y,i\rangle$, first check the encoding validity and reject invalid inputs.
	Let $F=G_N(y)$.
	Use an~$\NP$ oracle to compute the lexicographically first satisfying assignment of~$F$, if one exists.
	First ask whether $F$ is satisfiable; if not, output~$0$.
	Otherwise, build the assignment bit by bit.
	For a~prefix $p\in\{0,1\}^j$, the oracle query asks whether
	\[\exists a\in\{0,1\}^N\quad a\text{ extends }p\text{ and }F(a)=1.\]
	Using at most $N+1$ such oracle queries, one obtains the lexicographically first satisfying assignment.
	Since $N=\operatorname{poly}(\ell)$, this is a~polynomial-time computation with an~$\NP$ oracle.
\end{proof}

Now, the idea is the following: if $\problemabbr{LexWit}_G$ has small circuits on input length $\ell$, then fixing $N$ and a~hard seed~$y$ inside such a~circuit would leave a~small circuit whose only input is the index~$i$.
That circuit would output the satisfying assignment bit $a_i$, contradicting \Cref{lem:generated-hard-witness}.

\begin{theorem}
\label{thm:hard-generator-pnp-size}
	Fix $0<\alpha\le 1$.
	Assume $\cc{ETH}$, and suppose that an~$\alpha$-compressed hard formula generator $G=\{G_N\}$ exists.
	Let $q\colon\mathbb N\to\mathbb R_{>0}$ be any polynomial-time computable function such that $q(m)=\omega(\log m)$ and
	\[\frac{m^{1/\alpha}}{q(m)} = \omega(1).\]
	Then
	\[\P^{\NP}\nsubseteq\cc{SIZE}\left[\frac{m^{1/\alpha}}{q(m)}\right].\]
\end{theorem}
\begin{proof}
	Recall that $\ell(N)=\Theta(N^\alpha)$ is the input length of $\problemabbr{LexWit}_G$ corresponding to formulas on $N$ variables.
	Apply Lemma~\ref{lem:generated-hard-witness} with
	\[s(N)=\left\lfloor\frac{N}{\sqrt{q(\ell(N))\log N}}\right\rfloor.\]
	This threshold is polynomial-time computable, and the lemma applies since $q(\ell(N))=\omega(\log N)$ implies
	\[s(N)\log(s(N)+\log N)=O\left(\frac{N\log N}{\sqrt{q(\ell(N))\log N}}\right)=o(N).\]
	Hence, for infinitely many $N$, there is a~seed $y_N\in\{0,1\}^{r_\alpha(N)}$ such that $F_N=G_N(y_N)$ is satisfiable, but every satisfying assignment of $F_N$ has circuit complexity greater than $s(N)$.
	Let $a_N$ be the lexicographically first satisfying assignment of~$F_N$.

	Suppose, toward contradiction, that $\problemabbr{LexWit}_G\in\cc{SIZE}\left[m^{1/\alpha}/q(m)\right]$.
	For all sufficiently large $N$ from the infinite hard set, take a~circuit for $\problemabbr{LexWit}_G$ on inputs of length $\ell(N)$ and hardwire the bits encoding $N$ and~$y_N$.
	The resulting circuit takes only the index $i$ as input and computes the string $a_N$.
	Its size is
	\[O\left(\frac{\ell(N)^{1/\alpha}}{q(\ell(N))}\right)=O\left(\frac{N}{q(\ell(N))}\right)=o\left(\frac{N}{\sqrt{q(\ell(N))\log N}}\right)=o(s(N)).\]
	For all sufficiently large hard $N$, this contradicts $\cc{CC}[a_N]>s(N)$.
\end{proof}

\begin{restatable}{theorem}{generatedFormulaPNPSuperlinear}
	Fix $0<\eps<1$.
	Assume $\cc{ETH}$, and suppose there is a~uniform family of maps
	\[G_N\colon\{0,1\}^{N^{1-\eps}}\to\{\text{$3$-CNF formulas on variables }x_1,\ldots,x_N\}\]
	such that $G_N(y)$ is computable in time $\operatorname{poly}(N)$, and \problemabbr{SAT} restricted to $\operatorname{Im}(G_N)$ has no $2^{o(N)}$-time algorithm.
	Then, for every constant $0<\delta<\frac{\eps}{1-\eps}$,
	\[\P^{\NP}\not\subseteq\cc{SIZE}\left[m^{1+\delta}\right].\]
\end{restatable}
\begin{proof}
	Apply Theorem~\ref{thm:hard-generator-pnp-size} with $\alpha=1-\eps$.
	Let
	\[\gamma\coloneqq\frac{1}{1-\eps}-(1+\delta)>0.\]
	Choose the function $q(m)=m^{\gamma/2}$.
	Then $q(m)=\omega(\log m)$ and
	\[\frac{m^{1/\alpha}}{q(m)}=m^{1+\delta+\gamma/2}.\]
	Theorem~\ref{thm:hard-generator-pnp-size} gives a~language $\problemabbr{LexWit}_G\in \P^{\NP}$ that is not in $\cc{SIZE}[m^{1+\delta+\gamma/2}]$.
\end{proof}

\section{Open Problems}
\label{sec:conclusion}

The~results of this paper show that target degeneracy has a~genuinely intermediate behavior: it~is too rich for the~known algorithms, but the~available lower-bound machinery still leaves the~polynomial-target regime unresolved.
\Cref{thm:main-lb-formal} shows that the~dependence on $\operatorname{degen}(H)$ must be linear in the~exponent whenever this parameter grows, while \Cref{thm:constant-deg-lb-formal} shows that even degeneracy~$2$ can be hard when the~target is quasi-polynomial.
At the~same time, \Cref{thm:unified-nocomp-formal} suggests that substantially stronger lower bounds are unlikely.

The~most direct algorithmic question is whether the~trivial $O^*(h^n)=2^{O(n\log h)}$ upper bound can be improved using degeneracy.
Is there an~algorithm solving $\problemabbr{HOM}$ in time $O^*\left(2^{\operatorname{poly}(\operatorname{degen}(H))\cdot n}\right)$?
Such an~algorithm would match \Cref{thm:main-lb-formal} up to the~dependence on $\operatorname{degen}(H)$ and would separate degeneracy from target parameters that admit $p(H)^{O(n)}$ algorithms.

The~first concrete case is bounded degeneracy.
For every fixed $d\ge2$, can $\problemabbr{HOM}$ on instances with $\operatorname{degen}(H)\le d$ and $h\le\operatorname{poly}(n)$ be solved in time $O^*\left(2^{O_d(n)}\right)$?
\Cref{thm:constant-deg-lb-formal} excludes this only when quasi-polynomial targets are allowed, and targets of degeneracy at most~$1$ are forests and hence admit plain-exponential algorithms.
Thus the~question is open already for $d=2$ and polynomial-size targets.

A~second direction is to understand whether the~no-compression hypothesis has its own unconditional evidence.
\Cref{app:lookup-barrier} explains why the~lookup-table argument behind the~barrier of~Kulikov and Mihajlin~\cite{KM24} does not directly give such evidence for \cc{ETH}.
In short, their argument turns a~claimed $2^{\Omega(s)}$ lower bound into a~uniform $2^{(1-\Omega(1))N}$ algorithm for the $N$-variable \problemabbr{SAT} problem, by branching on all but a~constant fraction of the~variables and answering the~remaining short queries from a~precomputed table.
This contradicts \cc{SETH}, but it does not contradict \cc{ETH}, which rules out only $2^{o(N)}$ algorithms.
In the~spirit of~\cite{KM24}, can one nevertheless prove a~barrier showing that a~disproof of \Cref{hyp:nocomp} would imply an~unexpected consequence?
One concrete form of this question is whether there is a~generator of hard $3$-$\problemabbr{SAT}$ formulas with a~linear number of bits\footnote{This question was raised by Daniel Lokshtanov, personal communication.}.

\section*{Acknowledgments}
The authors used ChatGPT~5.5 to~assist in~strengthening the constructions in~\Cref{sec:linear-lb} and in~checking the lower-bound calculations in~\Cref{sec:generated-hard-witnesses}.
The resulting suggestions were reviewed by~the authors and incorporated only after independent verification.
The authors assume responsibility for all content.

\appendix

\section{Bounded Chromatic Number}
\label{sec:bounded-chi-poly-target}

This section notes that the bounded chromatic number lower bound of~\cite{FGKM15} can be strengthened via ideas in~\cite{CFGKMPS17}.

\begin{restatable}{theorem}{thmConstChromatic}
    \label{thm:bounded-chi-poly-target}
    Unless $\cc{ETH}$ fails, for every constant $\alpha>0$ there is a~constant $\rho=\rho(\alpha)>0$ such that no~algorithm decides $\problemabbr{HOM}$ in~time
    \[O^*\left(n^{\rho \cdot n}\right)\]
    on~instances $(G,H)$ satisfying $h\le n^\alpha$ and $\chi(H)\le15$.
\end{restatable}

We state four lemmas and then combine them.
We use a~grouping of~\cite{CFGKMPS17} and then use the chromatic-number reduction of~\cite{FGKM15} to get the constant bound on~$\chi(H)$.

\begin{lemma}[{\cite[Lemma~3]{CFGKMPS17}}]
\label{lem:cfgkmps-grouping}
	For any constant $d$, there is a~constant $\lambda=\lambda(d)$ and a~polynomial-time algorithm with the following property.
	Given an~$n$-vertex graph $Q$ of maximum degree at most~$d$ and an~integer $r$ satisfying $2\le r\le\sqrt{n/(2\lambda)}$, the algorithm constructs a~grouping graph $\mathcal B$ and a~labeling $\ell_{\mathcal B}\colon V(\mathcal B)\to[\lambda r]$ such that
	\[|V(\mathcal B)|\le n/r,\]
	$\ell_{\mathcal B}$ is a~proper coloring of $\mathcal B^2$, every bucket $B\in V(\mathcal B)$ is an~independent set in~$Q$, and between any two buckets there is at most one edge of~$Q$.
\end{lemma}

\begin{lemma}[{\cite[Lemma~4]{CFGKMPS17}}]
\label{lem:cfgkmps-list-encoding}
	Let $Q$ be a~graph and let $\Lambda$ be a~grouping graph of~$Q$ with label set $[L]$.
	Suppose that the labeling is a~proper coloring of $\Lambda^2$, every bucket is an~independent set in~$Q$, and between any two buckets there is at most one edge of~$Q$.
	Then one can construct, in time polynomial in the output size, a~list-homomorphism instance $(\Lambda,\Gamma,\mathcal L)$ such that
	\[Q\text{ is }3\text{-colorable}\quad\Longleftrightarrow\quad(\Lambda,\Gamma,\mathcal L)\text{ is satisfiable},\]
	and
	\[|V(\Gamma)|\le L\cdot 4^L.\]
\end{lemma}

\begin{lemma}[{\cite[Lemma~4]{FGKM15}}]
\label{lem:fgkm-target-recoloring}
	Given a~list-homomorphism instance $(\Lambda,\Gamma,\mathcal L)$ and a~proper $k$-coloring of~$\Lambda$, one can construct in polynomial time an~equivalent list-homomorphism instance $(\Lambda,\Gamma',\mathcal L')$ with
	\[|V(\Gamma')|=k|V(\Gamma)|\qquad\text{and}\qquad\chi(\Gamma')\le k.\]
\end{lemma}

\begin{lemma}[{\cite[Lemma~5]{FGKM15}}]
\label{lem:fgkm-list-removal}
	Given a~list-homomorphism instance $(\Lambda,\Gamma,\mathcal L)$ with $|V(\Lambda)|=n$, $|V(\Gamma)|=h$, and $\chi(\Gamma)\le t$, one can construct in polynomial time an~equivalent ordinary homomorphism instance $(G,H)$ satisfying
	\[|V(G)|\le n+(h+1)(t+11),\qquad |V(H)|\le(h+1)(t+11),\qquad \chi(H)\le t+10.\]
\end{lemma}

\begin{corollary}
\label{lem:bounded-chi-list-compression}
	There are constants $\lambda\ge1$ and $\sigma_0>1$, and a~polynomial-time algorithm with the following property.
	Given an~$n$-vertex graph $Q$ of maximum degree at most~$4$ and an~integer $2\le r\le\sqrt{n/(2\lambda)}$, the algorithm constructs a~list-homomorphism instance $(\Lambda,\Gamma,\mathcal L)$ such that
	\[Q\text{ is }3\text{-colorable}\quad\Longleftrightarrow\quad(\Lambda,\Gamma,\mathcal L)\text{ is satisfiable},\]
	and
	\[|V(\Lambda)|\le\frac{5n}{r},\qquad\chi(\Lambda)\le5,\qquad|V(\Gamma)|\le\sigma_0^r.\]
\end{corollary}
\begin{proof}
	Let $q\colon V(Q)\to[5]$ be a~greedy proper $5$-coloring of~$Q$.
	Apply \Cref{lem:cfgkmps-grouping} with $d=4$, and let $\mathcal B$ be the resulting grouping graph with labeling $\ell_{\mathcal B}:V(\mathcal B)\to[L]$, where $L\le\lambda r$.

	Refine every bucket $B\in V(\mathcal B)$ by the color classes of~$q$: for each $a\in[5]$ with $B_a=B\cap q^{-1}(a)\ne\varnothing$, create a~bucket $B_a$.
	Let $\Lambda$ be the graph on these refined buckets, with two refined buckets adjacent exactly when an~edge of~$Q$ runs between them, and let each refined bucket inherit the label of its parent bucket.
	Then $|V(\Lambda)|\le5|V(\mathcal B)|\le5n/r$, and $q$ gives a~proper $5$-coloring of~$\Lambda$.

	Every refined bucket is independent, and between any two refined buckets there is at most one edge of~$Q$.
	The inherited labeling remains a~proper coloring of $\Lambda^2$: different parent buckets inherit this from $\mathcal B^2$, while two buckets with the same parent are not adjacent and cannot have a~common neighbor without creating two edges between one pair of parent buckets.

	Apply \Cref{lem:cfgkmps-list-encoding} to the refined grouping $\Lambda$.
	It gives an~equivalent list-homomorphism instance $(\Lambda,\Gamma,\mathcal L)$ with
	\[|V(\Gamma)|\le L\cdot 4^L\le\sigma_0^r\]
	for a~constant $\sigma_0$ depending only on~$\lambda$.
\end{proof}

\begin{proof}[Proof of Theorem~\ref{thm:bounded-chi-poly-target}]
	Fix $\alpha>0$, and let $\beta>0$ be the bounded-degree \problemname{$3$-Coloring} constant from \Cref{thm:bdcolor}.
	Let $\lambda$ and $\sigma_0$ be the constants from \Cref{lem:bounded-chi-list-compression}.

	Choose a~constant $\sigma\ge\max\{\sigma_0,2\}$ large enough that the following consequence holds for every $r\ge2$.
	If \Cref{lem:bounded-chi-list-compression} is followed by \Cref{lem:fgkm-target-recoloring} with $k=5$ and then by \Cref{lem:fgkm-list-removal} with $t=5$, the resulting ordinary homomorphism instance $(G_0,H_0)$ is equivalent to $3$-coloring of~$Q$ and satisfies
	\[|V(G_0)|\le\frac{5n}{r}+\sigma^r,\qquad |V(H_0)|\le\sigma^r,\qquad \chi(H_0)\le15.\]
	Indeed, the target recoloring changes the list target size by a~factor~$5$, and list removal adds only $(h+1)(5+11)$ vertices to each side.

	Set
	\[\delta=\frac{\min\{1,\alpha\}}{4\log \sigma}.\]
	Given an~$n$-vertex maximum-degree-$4$ graph $Q$, put $r=\lfloor\delta\log n\rfloor$.
	For all sufficiently large $n$, the condition $2\le r\le\sqrt{n/(2\lambda)}$ holds, so Lemma~\ref{lem:bounded-chi-list-compression} applies.

	For the resulting ordinary homomorphism instance,
	\[\sigma^r\le \sigma^{\delta\log n}=n^{\delta\log \sigma}\le n^{1/4}.\]
	Therefore, for all sufficiently large $n$,
	\[|V(G_0)|\le\frac{5n}{r}+n^{1/4}\le\frac{6n}{r}.\]

	Pad $G_0$ with isolated vertices so that the final source graph $G$ has exactly
	\[N=\left\lceil\frac{6n}{r}\right\rceil\]
	vertices, and set $H:=H_0$.
	Adding isolated vertices does not affect the existence of a~homomorphism, so
	\[Q\text{ is }3\text{-colorable}\quad\Longleftrightarrow\quad G\to H.\]
	The chromatic-number bound remains $\chi(H)\le15$.

	The target size is polynomial in the final source size.
	Indeed,
	\[|V(H)|\le \sigma^r\le n^{\alpha/4}.\]
	Since $N=\Theta(n/\log n)$, we have $n^{\alpha/4}\le N^\alpha$ for all sufficiently large $n$.
	Thus the constructed instances satisfy
	\[|V(G)|=N,\qquad |V(H)|\le N^\alpha,\qquad \chi(H)\le15.\]

	Now suppose that, for some constants $\rho>0$ and $d\ge0$, this restricted class of $\problemabbr{HOM}$ instances can be decided in time
	\[O\bigl(N^{\rho N}\cdot|V(H)|^d\bigr).\]
	For all sufficiently large $n$, we have $r\ge(\delta/2)\log n$, and hence
	\[N\le\frac{7n}{r}\le\frac{14n}{\delta\log n}.\]
	Since also $\log N\le\log n$, it follows that
	\[N\log N\le\frac{14}{\delta}n.\]

	As $|V(H)|\le N^\alpha$, the assumed algorithm would solve \problemname{$3$-Coloring} in time
	\[2^{\rho N\log N}N^{O(1)}\le2^{(14\rho/\delta+o(1))n}.\]
	Choose $\rho=\beta\delta/28$.
	Then the last running time is $2^{(\beta/2+o(1))n}$, contradicting \Cref{thm:bdcolor}.
\end{proof}

\section{Graph Parameters} \label{section:parameters}

This appendix defines all target-side graph parameters from~\Cref{fig:parameter-dominance-map} and proves the implications represented by its arrows.

\begin{itemize}
	\item The chromatic number $\chi(G)$ is the least number of colors in a~proper vertex coloring of~$G$.

	\item The maximum degree $\Delta(G)$ is the maximum degree of a~vertex of~$G$.

	\item The degeneracy $\operatorname{degen}(G)$ is the least integer $d$ such that every non-empty subgraph of~$G$ has a~vertex of degree at most~$d$.

	\item A~tree decomposition of~$G$ is a~tree $T$ with bags $\beta(x)\subseteq V(G)$, $x\in V(T)$, such that every vertex and every edge of~$G$ is covered by a~bag and the bags containing any fixed vertex form a~non-empty connected subtree of~$T$.
		The treewidth $\operatorname{tw}(G)$ is the minimum value of $\max_{x\in V(T)}|\beta(x)|-1$ over all tree decompositions of~$G$~\cite{RS86}.

	\item For $k\ge1$, a~$k$-expression builds $k$-labeled graphs by creating one vertex of a~chosen label, taking disjoint unions, joining all vertices of one label to all vertices of another label, and relabeling vertices.
		The clique-width $\operatorname{cw}(G)$ is the least $k$ such that some $k$-expression represents $G$ after labels are forgotten~\cite{CER93,CO00}.

	\item The extended clique-width $\operatorname{ecw}(G)$ is the analogous minimum for extended $k$-expressions of Bulatov and Dadsetan~\cite{BD20}.
		Such expressions add to $k$-expressions the connect operation $\eta_{\mathcal T}$ between two labeled graphs and the beta operation $\beta_{\vec n,\sigma,\mathcal S}$, which replaces each vertex of label~$i$ by $n_i+1$ labeled copies and keeps edges according to~$\mathcal S$.

	\item For a~linear order $v_1,\ldots,v_m$ of $V(G)$, its bandwidth is $\max\{|i-j|:v_iv_j\in E(G)\}$.
		The bandwidth $\operatorname{bw}(G)$ is the minimum such value over all linear orders, and the complement bandwidth is $\operatorname{bw}(\overline{G})$~\cite{Harper64}.

	\item A~track layout of~$G$ is a~proper coloring $\gamma$ together with a~linear order $\prec$ of~$V(G)$ such that no two edges $u_1u_2,v_1v_2$ with $\gamma(u_1)=\gamma(v_1)$ and $\gamma(u_2)=\gamma(v_2)$ satisfy $u_1\prec v_1$ and $v_2\prec u_2$.
		The track number $\operatorname{tn}(G)$ is the least number of colors in~a~track layout~\cite{DPW04}.

	\item Given a~proper coloring $\gamma\colon V(G)\to[k]$, let $\mathbf{\Gamma}_{G,\gamma}$ be the binary constraint language on domain $V(G)$ with relations
		\[R_{ij}=\{(u,v)\in V(G)^2\colon\gamma(u)=i,\ \gamma(v)=j,\ uv\in E(G)\},\qquad i,j\in[k].\]
		A~ternary operation $m\colon D^3\to D$ is a~majority operation if $m(b,a,a)=m(a,b,a)=m(a,a,b)=a$ for all $a,b\in D$.
		For tuples $t_1,t_2,t_3$ of the same arity, write $f(t_1,t_2,t_3)$ for the tuple obtained by applying $f$ coordinatewise.
		A~persistent majority triple for a~binary crisp constraint language $\mathbf{\Gamma}$ over domain~$D$ is a~triple $(f_1,f_2,f_3)$ of operations $D^3\to D$ such that $f_1$ is majority,
		\[\{\!\{f_1(a_1,a_2,a_3),f_2(a_1,a_2,a_3),f_3(a_1,a_2,a_3)\}\!\}=\{\!\{a_1,a_2,a_3\}\!\}\]
		for all $a_1,a_2,a_3\in D$, and
		\[\{\!\{f_1(t_1,t_2,t_3),f_2(t_1,t_2,t_3),f_3(t_1,t_2,t_3)\}\!\}=\{\!\{t_1,t_2,t_3\}\!\}\]
		for every relation $R\in\mathbf{\Gamma}$ and all $t_1,t_2,t_3\in R$.
		A~persistent majority coloring of~$G$ is a~proper coloring $\gamma$ such that $\mathbf{\Gamma}_{G,\gamma}$ has a~persistent majority triple.
		The persistent majority number $\operatorname{pmn}(G)$ is the least number of colors in~a~persistent majority coloring~\cite{Carbonnel26}.
\end{itemize}

\subsection{Implications in the Parameter Map}
\begin{lemma}
	For every finite graph $G$, the following implications hold:
	\begin{align*}
		\operatorname{cw}(G) &\le 2^{\operatorname{tw}(G) + 2}, & \operatorname{cw}(G) &\le 2^{\operatorname{bw}(\overline{G}) + 3}, \\
		\operatorname{ecw}(G) &\le \operatorname{cw}(G), & \operatorname{tn}(G) &\le f_{\operatorname{tw}}(\operatorname{tw}(G)), \\
		\operatorname{pmn}(G) &\le \operatorname{tn}(G), & \operatorname{pmn}(G) &\le \Delta(G)^2 + 1,
	\end{align*}
	for some function $f_{\operatorname{tw}}$.
\end{lemma}
\begin{proof}
	$\operatorname{cw}(G)\le 2^{\operatorname{tw}(G)+2}$ is due to Courcelle and Olariu~\cite{CO00}.
	Every $k$-expression is an~extended $k$-expression, hence $\operatorname{ecw}(G)\le\operatorname{cw}(G)$.

	Let $b=\operatorname{bw}(\overline{G})$.
	Choose an~ordering $v_1,\ldots,v_m$ of $V(G)$ witnessing bandwidth at most~$b$ in~$\overline{G}$.
	Then the bags $B_i=\{v_j\colon i\le j\le i+b\}$ form a~path decomposition of~$\overline{G}$ of width at most~$b$, because every edge of~$\overline{G}$ has both endpoints in~one such interval.
	Therefore $\operatorname{tw}(\overline{G})\le b$, and the previous paragraph gives $\operatorname{cw}(\overline{G})\le 2^{b+2}$.
	Clique-width is bounded under complementation, more precisely $\operatorname{cw}(\overline{F})\le2\operatorname{cw}(F)$ for every graph $F$~\cite{CO00}.
	Applying this to $F=\overline{G}$ yields $\operatorname{cw}(G)\le2^{b+3}$, and then $\operatorname{ecw}(G)\le2^{b+3}$.

	Graphs of bounded treewidth have bounded track number due to~\cite{DMW05}.
	On the other hand, bounded track number or bounded maximum degree imply bounded persistent majority number due to~\cite{Carbonnel26}.
\end{proof}

\begin{lemma}
    For every graph~$G$,
    \[\operatorname{degen}(G)\le 4(\operatorname{pmn}(G)-1).\]
\end{lemma}
\begin{proof}
    Let $k=\operatorname{pmn}(G)$, and let $\gamma\colon V(G)\to[k]$ be a~persistent majority coloring witnessed by a~persistent majority triple $(f_1,f_2,f_3)$.
    For colors $i<j$, let $B_{ij}$ be the bipartite graph induced by the edges between color classes $\gamma^{-1}(i)$ and $\gamma^{-1}(j)$.

    We first show that each $B_{ij}$ is $2$-degenerate.
    Suppose to the contrary that some subgraph $B$ of~$B_{ij}$ has minimum degree at least~$3$.
    Choose a~vertex $x$ on the $i$-side of~$B$ with three distinct neighbors $y_1,y_2,y_3$ on the $j$-side.
    For each $r\in\{1,2,3\}$, choose a~neighbor $x_r\ne x$ of~$y_r$ in~$B$.

    Apply $f_1$ to the three tuples $(x,y_1),(x,y_2),(x_3,y_3)\in R_{ij}$.
    On the first coordinate, $f_1(x,x,x_3)=x$.
    Since the output tuple $f_1((x,y_1),(x,y_2),(x_3,y_3))$ is one of~the three input tuples, it must be $(x,y_1)$ or $(x,y_2)$.
    Hence $f_1(y_1,y_2,y_3)\in\{y_1,y_2\}$.
    The same argument applied to $(x,y_1),(x_2,y_2),(x,y_3)$ gives $f_1(y_1,y_2,y_3)\in\{y_1,y_3\}$, and applied to $(x_1,y_1),(x,y_2),(x,y_3)$ gives $f_1(y_1,y_2,y_3)\in\{y_2,y_3\}$.
    These three sets have empty intersection, a~contradiction.
    Thus each $B_{ij}$ is $2$-degenerate.

    Now let $U\subseteq V(G)$ and set $U_i=U\cap\gamma^{-1}(i)$.
    For each pair $i<j$, the induced bipartite graph between $U_i$ and $U_j$ is $2$-degenerate, and therefore has at most $2(|U_i|+|U_j|)$ edges.
    Because $\gamma$ is proper, all edges of~$G[U]$ lie between distinct color classes, so
    \[\left|E(G[U])\right|\le \sum_{1\le i<j\le k}2(|U_i|+|U_j|)=2(k-1)|U|.\]
    Every induced subgraph of~$G$ has average degree at most $4(k-1)$.
    Therefore every subgraph of~$G$ has a~vertex of degree at most $4(k-1)$, as required.
\end{proof}

\subsection{Padding Lower Bounds}
We also note that the upper bounds in~\Cref{fig:parameter-dominance-map} are tight under $\cc{ETH}$ for all parameters (except complement bandwidth).
\begin{observation}\label{obs:parameter-padding-lower-bound}
	Let $\mu$ be a~graph parameter such that $\mu(H)\le |V(H)|$ for every graph~$H$, and suppose that $\mu$ is unchanged by adding isolated vertices.
	Let $p=p(n)$ and $h=h(n)$ be integer-valued functions with $3\le p(n)\le h(n)$.
	Unless $\cc{ETH}$ fails, there is no algorithm that solves $\problemabbr{HOM}$ on all instances with $|V(G)|=n$, $|V(H)|=h(n)$, and $\mu(H)\le p(n)$ in~time $p(n)^{o(n)}$.
\end{observation}
\begin{proof}
	By the lower bound of~\cite{CFGKMPS17}, under $\cc{ETH}$, $\problemabbr{HOM}$ on instances with $|V(G)|=n$ and $|V(H_0)|=p(n)$ has no $p(n)^{o(n)}$-time algorithm.
	Given such an~instance $(G,H_0)$, add $h(n)-p(n)$ isolated vertices to~$H_0$ and call the resulting target~$H$.
	Then $|V(H)|=h(n)$ and, by the assumption on~$\mu$,
	\[\mu(H)=\mu(H_0)\le |V(H_0)|=p(n).\]
	Moreover, $G\to H$ if and only if $G\to H_0$.
	Indeed, dummy target vertices are isolated, so any source vertex mapped to one of them is isolated in~$G$ and can be remapped to an~arbitrary vertex of~$H_0$.
	Thus a~$p(n)^{o(n)}$-time algorithm for the padded instances would give a~$p(n)^{o(n)}$-time algorithm for the hard instances of~\cite{CFGKMPS17}, a~contradiction.
\end{proof}

Complement bandwidth $\operatorname{bw}(\overline{H})$ increases when adding isolated vertices, so the observation above does not follow.
Moreover, for any graph $H$ on $h$ vertices with $\operatorname{bw}(\overline{H}) \le p$ it follows that $\omega(H) \ge \left\lceil \frac{h}{p + 1} \right\rceil$, hence if $h(n) > (n - 1) (p(n) + 1)$, then for every graph $G$ on $n$ vertices, there is a~homomorphism $G \to H$.
Thus, the above observation is false for complement bandwidth.

\section{The Graph Homomorphism Problem and Its Complexity Dichotomies}\label{sec:fine-grained-dichotomy}

The lower bound $2^{\Omega(n\log h)}$ for $\problemabbr{HOM}$ holds for the general case when
both graphs $G$~and~$H$ are part of~the input and
are not restricted in~any way~\cite{CFGKMPS17}. At~the same time, in~various applications
of the~\problemname{Graph Homomorphism} problem, one or~both of~these graphs
may come from a~restricted graph family or~even be~fixed.
There~is a~rich variety of~results studying the complexity
of~special cases of~\problemabbr{HOM}.
Complexity dichotomies ask where the boundary lies between the tractable and intractable sides of the relevant complexity model, such as $\P$ versus $\NP$-hardness, $\cc{FPT}$ versus $\cc{W}[1]$-hardness, or $\cc{FP}$ versus $\#\P$-hardness.
Fine-grained classifications ask, once the problem is intractable, which structural features determine the correct exponential running time.

\paragraph{Fixed and Restricted Source Graphs.}
For the~\emph{fixed source} problem $\problemabbr{HOM}(G, \cdot)$, the brute-force algorithm with running time $O(h^{n})$ is~polynomial.
For \emph{restricted source} graph classes~$\mathcal{G}$,
the question~is what properties of~$\mathcal{G}$
make $\problemabbr{HOM}(\mathcal G,\cdot)$ tractable.
Dalmau, Kolaitis, and Vardi~\cite{DKV02} and Grohe~\cite{Grohe07} showed that $\problemabbr{HOM}(\mathcal{G}, \cdot)$, parametrized by~the size of~the source graph, is in $\cc{FPT}$ precisely when the cores of~graphs in~$\mathcal{G}$ have bounded treewidth, and is $\cc{W}[1]$-hard otherwise.
See \Cref{figure:restricted-source-dichotomy}.
Dalmau and Jonsson~\cite{DJ04} proved the analogous result for counting homomorphisms.

\tikzset{p/.style={rectangle, rounded corners, draw=black, minimum height=6mm}}

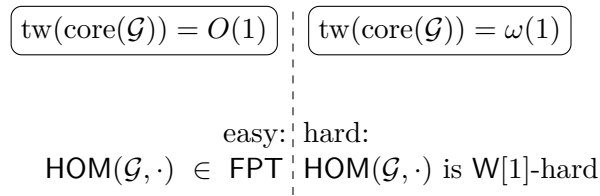
\begin{figure}[ht]
	\begin{center}
		\begin{tikzpicture}
			\draw[dashed] (0, .75) -- (0, -1.75);
			\node[p, left] at (-0.2, 0.5) {$\operatorname{tw}(\operatorname{core}(\mathcal{G}))=O(1)$};
			\node[p, right] at (0.2, 0.5) {$\operatorname{tw}(\operatorname{core}(\mathcal{G}))=\omega(1)$};
			\node[below left, text width=40mm, align=right] at (0, -.5) {\strut easy:\\ $\problemabbr{HOM}(\mathcal G,\cdot) \in \cc{FPT}$};
			\node[below right, text width=40mm] at (0, -.5) {\strut hard:\\ $\problemabbr{HOM}(\mathcal G,\cdot)$ is $\cc{W}[1]$-hard};
		\end{tikzpicture}
	\end{center}
	\caption{Parameterized complexity dichotomy for the source-restricted problem $\problemabbr{HOM}(\mathcal G,\cdot)$.}
	\label{figure:restricted-source-dichotomy}
\end{figure}

Fine-grained source-side results further show that the known dynamic-programming exponents are often the right ones.
For vertex cover,~\cite{FGKM15} give an~$O(h^{\operatorname{vc}(G)})$ algorithm and a~matching $\cc{ETH}$ lower bound $h^{\Omega(\operatorname{vc}(G))}$.
Marx~\cite{Marx10} showed that the $h^{O(\operatorname{tw}(G))}$ algorithm cannot be significantly improved under $\cc{ETH}$.
Okrasa and Rz{\k{a}}{\.z}ewski~\cite{OR21} showed that if $H$ is a~projective core, then there is no $(h-\varepsilon)^{\operatorname{tw}(G)} n^{O(1)}$ algorithm, for any $\varepsilon>0$, under $\cc{SETH}$.
Since almost all graphs are projective cores~\cite{HN92,LN04}, this covers almost all targets; moreover, the result can be lifted to all targets under long-standing conjectures on~the properties of~projective cores~\cite{Larose02,LT01}.
This was extended by~\cite{GHKOS24}, who proved both lower and upper bounds $O^*(s(H)^{\operatorname{cw}(G)})$ under the same conjectures, where $s$ is a~signature number of~$H$.
Okrasa, Piecyk, and Rz{\k{a}}{\.z}ewski~\cite{OPR20}, building on~the result of Egri, Marx, and Rz{\k{a}}{\.z}ewski~\cite{EMR18}, showed that for all graphs~$H$ where \problemname{List Homomorphism} (\problemabbr{LHOM}) is $\NP$-hard, it also cannot be solved in time $(h - \varepsilon)^{\operatorname{tw}(G)} n^{O(1)}$,
for any $\varepsilon > 0$, under $\cc{SETH}$.
For counting list homomorphisms, Focke, Marx, and Rz{\k{a}}{\.z}ewski~\cite{FMR24} identify the exact exponent base in running time and prove matching $\#\cc{SETH}$ lower bounds.
\cite{GMNPR24} provide a~characterization of the case when $G$~is parametrized by cutwidth in terms of~the asymptotic rank parameter they introduce.

\paragraph{Fixed Target Graphs.}
The \emph{fixed target} version $\problemabbr{HOM}(\cdot,H)$ is~also well understood.
Hell and Ne\v{s}et\v{r}il~\cite{HN90} proved that $\problemabbr{HOM}(\cdot,H)$ is polynomial-time solvable when $H$ is bipartite, and is $\NP$-complete otherwise; Meyer and Opr\v{s}al recently gave a~topological proof~\cite{MO25}.
	Thus, the coarse-grained complexity class of~$\problemabbr{HOM}(\cdot,H)$ is determined by~an~easy-to-compute parameter of~$H$: whether $\chi(H)\le2$ or~not; see \Cref{figure:fixed-target-dichotomy}.
	The Feder--Vardi dichotomy conjecture for finite-domain $\problemabbr{CSP}$s~\cite{FV98} was proved independently by Bulatov and Zhuk, giving a~complete algebraic classification of fixed-target tractability~\cite{Bulatov17,Zhuk17}.

	\begin{figure}[ht]
		\begin{center}
			\begin{tikzpicture}
				\draw[dashed] (0, .75) -- (0, -1.75);
				\node[p, left] at (-.2, 0.5) {$\chi(H) \le 2$};
			\node[p, right] at (.2, 0.5) {$\chi(H) > 2$};
			\node[below left, text width=40mm, align=right] at (0, -.5) {\strut easy:\\ $\problemabbr{HOM}(\cdot, H) \in \P{}$};
			\node[below right, text width=40mm] at (0, -.5) {\strut hard:\\ $\problemabbr{HOM}(\cdot, H)$ is \NP{}-hard};
		\end{tikzpicture}
	\end{center}
	\caption{Complexity dichotomy for the fixed-target problem $\problemabbr{HOM}(\cdot,H)$.}
	\label{figure:fixed-target-dichotomy}
\end{figure}
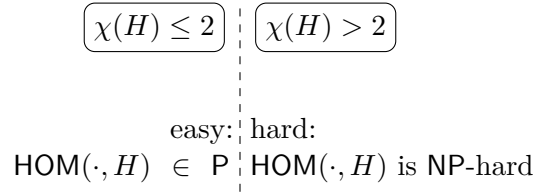

The same dichotomy philosophy extends to several homomorphism variants.
Feder and Hell~\cite{FH98} showed that if each vertex of $H$ has a~loop, then $\problemabbr{LHOM}(\cdot, H)$ is in $\P$ if $H$ is an~interval graph and $\NP$-complete otherwise.
Feder, Hell, and Huang~\cite{FHH99} showed that if $H$ has no loops, then $\problemabbr{LHOM}(\cdot, H)$ is in~\P{} if the complement of $H$ is a~circular-arc graph with clique-cover number two, and is \NP{}-complete otherwise.
Finally, Feder, Hell, and Huang~\cite{FHH03} defined a~new class of graphs with possible loops, called bi-arc graphs, and showed that if $H$ is a~bi-arc graph, then $\problemabbr{LHOM}(\cdot, H)$ can be solved in polynomial time, and otherwise the problem is \NP{}-complete.
In~the reflexive case, bi-arc graphs coincide with interval graphs; in~the irreflexive case, they coincide with bipartite graphs whose complement is a~circular-arc graph.
Dyer and Greenhill~\cite{DG00} showed that $\#\problemabbr{HOM}(\cdot, H)$ is \#\P{}-complete unless every connected component of~$H$ is an~isolated vertex without a~loop, a~complete graph with all loops present, or a~complete unlooped bipartite graph; otherwise it is in $\cc{FP}$.
This dichotomy was extended in numerous ways: to~weighted graphs~\cite{BG05,GGJT10,GT11,CCL13,GCD23}, directed graphs~\cite{DGP07,CC17}, and the number of homomorphisms modulo a~fixed prime~\cite{FJ15,GGR14,GGR16}.
Chen, Curticapean, and Dell~\cite{CCD19} showed that the hard cases from \cite{DG00} require $2^{\Omega(n)}$ time under $\cc{ETH}$.
For finite-domain counting $\problemabbr{CSP}$s, the exact tractable side is also sharply classified by algebraic conditions~\cite{Bulatov13,DR13,FGZ19}.
We refer to Bulatov's survey~\cite{Bulatov18} and the book of Hell and Ne\v{s}et\v{r}il~\cite{HN26} for a~broader background.

\paragraph{Restricted Target Graphs.}
The fine-grained picture becomes much less complete for the \emph{restricted target} version $\problemabbr{HOM}(\cdot,\mathcal{H})$.
Several recent results identify properties of~the target class~$\mathcal{H}$ that bring $\problemabbr{HOM}(\cdot,\mathcal{H})$ down to single-exponential time.
Fomin, Heggernes, and Kratsch~\cite{FHK07} gave an~$O^*((2\operatorname{tw}(H)+1)^n)$ algorithm.
Wahlstr{\"o}m~\cite{Wahlstrom11} gave an~$O^*((2\operatorname{cw}(H)+1)^{\max\{n,h\}})$ algorithm for bounded clique-width.
Bulatov and Dadsetan~\cite{BD20} extended this to~bounded extended clique-width, giving an~$O^*((4\operatorname{ecw}+4)^h+(2\operatorname{ecw}+1)^n)$ algorithm.
Fomin, Golovnev, Kulikov, and Mihajlin~\cite{FGKM15} gave an~$O^*(\Delta(H)^n)$ algorithm for bounded maximum degree, while Rz{\k{a}}{\.z}ewski~\cite{Rzazewski14} gave an~$O^*((\operatorname{bw}(\overline H)+2)^n)$ algorithm for bounded complement bandwidth.
Recently, Carbonnel~\cite{Carbonnel26} proved a~$2^{O(n)}$ upper bound when $\mathcal{H}$ excludes a~fixed topological minor, and also gave $O^*(\operatorname{pmn}(H)^{n+h})$ and $O^*(\operatorname{tn}(H)^{n+h})$ algorithms.
On~the other hand, bounded chromatic number is not enough: \cite{FGKM15} rules out $2^{O(n)}$ algorithms for this case unless $\cc{ETH}$ fails.

\paragraph{The Fine-Grained Dichotomy for Homomorphism.}
A~natural target-side dream is to~find a~structural graph parameter $p$ that describes the plain-exponential frontier for $\problemabbr{HOM}(\cdot,\mathcal H)$.
Known algorithms place parameters such as bounded treewidth and excluded topological minors on~the algorithmic side, while bounded chromatic number is far from sufficient.
The weakest useful form of~such a~dichotomy would separate target families admitting a~$2^{O(n)}$ algorithm from those that require super-single-exponential time.
This raises the following question:
\begin{quote}
    Is~there
    a~natural graph parameter~$p$ such that bounded $p(\mathcal{H})$ implies
    that $\problemabbr{HOM}(\cdot, \mathcal{H}) \in 2^{O(n)}$ when $h \le n$, whereas
    unbounded $p(\mathcal{H})$ implies that there is~no~such
    algorithm under \cc{ETH}?
\end{quote}

Such fine-grained dichotomies are~known in~several settings.
For~example, Bringmann, Fischer, and~K{\"u}nnemann~\cite{BFK19} introduced a~hardness parameter~$H$ for~every $\exists^{k}\forall$ graph property~$\varphi$, and~showed that the~baseline $O(m^{k})$ algorithm for~deciding whether a~graph on~$m$ vertices satisfies~$\varphi$ can~be improved if~and only if~$H(\varphi) \leq 2$, assuming the~\textsc{Hyper-Clique Hypothesis}.
Berkholz, Keppeler, and~Schweikardt~\cite{BKS17} established a~dichotomy for~dynamic conjunctive-query answering, showing that it~can~be computed with~linear preprocessing time and~constant update time if~and only if the~homomorphic core of~the~query is~$q$-hierarchical.
Otherwise, the~size of~the~query result cannot be~maintained with~sublinear update time.
Moreover, Brand, Dell, and~Roth~\cite{BDR16} established a~dichotomy for~the~problem~$T_{x,y}$, where $(x,y) \in \mathbb{Q}^{2}$ is~a~fixed rational point.
Given a~graph~$G$, the~task is~to~evaluate the Tutte polynomial~\cite{Tutte54} at~the~point~$(x,y)$, that~is, to~compute $T_{x,y}(G)=T(G;x,y)$.
This problem belongs to~$\cc{FP}$ if~and only if $(x-1)(y-1)=1$.
For~every other point, it~is~$\#\P$-hard and~cannot be~solved in~time $2^{o(n)}$ under~$\#\cc{ETH}$~\cite{JVW90,DHMTW14,HT10,Curticapean18,BDR16}.
The~same work~\cite{BDR16} also established a~dichotomy for~the~problem of~counting satisfying assignments of~Boolean \problemabbr{CSP} instances over~a~constraint language~$\Gamma$.
The~problem belongs to~$\cc{FP}$ if~and only if every relation in~$\Gamma$ is~affine over~$\cc{GF}(2)$.
Otherwise, it~is~$\#\P$-complete and~cannot be~solved in~time $2^{o(n)}$ under~$\#\cc{ETH}$~\cite{CH96,BDR16}.

\subsection{Obstruction for Monotone Parameters}
We finish with a~simple obstruction that shows that any~monotone parameter cannot establish the dichotomy.

\begin{theorem}
	Let $p$ be an~unbounded graph parameter such that
	\[A\le_{\cc{ind}} H \qquad\Longrightarrow\qquad p(A)\le p(H),\]
	where $A\le_{\cc{ind}} H$ means that $A$ is an~induced subgraph of~$H$.
	Then there is a~target family~$\mathcal H$ with $\sup_{H\in\mathcal H}p(H)=\infty$ such that $\problemabbr{HOM}(\cdot,\mathcal H)$ on~inputs $(G,H)$ with $|V(G)|=n$, $H\in\mathcal H$, and $|V(H)|\le n$ is solvable in~time $2^{O(n)}$.

	The same conclusion holds if $p$ is monotone under subgraphs or~under minors.
\end{theorem}

\begin{proof}
	Since $p$ is unbounded, choose graphs $A_1,A_2,\ldots$ such that $p(A_j)\ge j$ for all~$j$.
	By taking a~subsequence of it we may assume that the sizes $a_j=|V(A_j)|$ are strictly increasing.

	For every integer $i\ge1$, let $j(i)$ be the largest index~$j$ such that $a_j\le i/2$, if such an~index exists.
	If no~such index exists, set $H_i=K_i$.
	If $j(i)$ exists, set
	\[H_i=K_{i-a_{j(i)}}\sqcup A_{j(i)}.\]
	Let $\mathcal H=\{H_i:i\ge1\}$.
	Then $|V(H_i)|=i$.
	In~the nontrivial case, $A_{j(i)}$ is an~induced subgraph of~$H_i$, so monotonicity gives
	\[p(H_i)\ge p(A_{j(i)})\ge j(i).\]
	As $i\to\infty$, also $j(i)\to\infty$, and hence $\sup_{H\in\mathcal H}p(H)=\infty$.

	It remains to~observe that every $H_i$ is homomorphically equivalent to~a~clique.
	This is clear when $H_i=K_i$.
	In~the nontrivial case, put $j=j(i)$ and $k_i=i-a_j$.
	Since $a_j\le i/2$, we have $k_i\ge a_j$.
	The clique component gives an~inclusion $K_{k_i}\to H_i$.
	Conversely, map the clique component of~$H_i$ identically to~$K_{k_i}$ and map $V(A_j)$ injectively into $V(K_{k_i})$.
	This is a~homomorphism because $K_{k_i}$ is complete and there are no~edges between the two components of~$H_i$.
	Thus $H_i$ and $K_{k_i}$ are homomorphically equivalent, and for every graph~$G$,
	\[G\to H_i \qquad\Longleftrightarrow\qquad G\to K_{k_i}.\]

	Now consider an~input $(G,H_i)$ with $|V(G)|=n$ and $i=|V(H_i)|\le n$.
	From $H_i$ one can find, in~polynomial time, a~complete connected component~$C$ of~size $\ell_i\ge i/2$.
	Then $K_{\ell_i}$ embeds into~$H_i$, and every other connected component has at~most $i-\ell_i\le\ell_i$ vertices, so it maps injectively to~$K_{\ell_i}$.
	Thus $H_i$ is homomorphically equivalent to~$K_{\ell_i}$, and the problem is exactly $\ell_i$-colorability of~$G$.
	Since $\ell_i\le i\le n$, it can be decided in~time $2^n n^{O(1)}$~\cite{BHK09}.
	Thus $\problemabbr{HOM}(\cdot,\mathcal H)$ in~the regime $h\le n$ is solvable in~time $2^{O(n)}$.

	The construction also proves the subgraph- and minor-monotone variants, since $A_{j(i)}$ is not only an~induced subgraph of~$H_i$, but also a~subgraph and a~minor of~$H_i$.
\end{proof}

Consequently, parameters such as $\operatorname{tw}(H)$, $\operatorname{degen}(H)$, $\Delta(H)$, $\chi(H)$, and $\operatorname{cw}(H)$ cannot give the hard direction of~the family-level question.

\section{SETH Version of the No-Compression Barrier}
\label{app:lookup-barrier}

For~the sake of~completeness, we~outline the~main idea behind the~result of~Kulikov and~Mihajlin~\cite{KM24}.

\begin{lemma}[$\cc{SETH}$ lookup-table barrier]
	Let $\mathcal A$ be a decision problem whose instances are encoded by $s$ bits, and
	assume that $\mathcal A$ can be solved in time $2^{O(s)}$.

	Then, if a~fine-grained reduction from \problemabbr{SAT} to $\mathcal A$ establishes, under $\cc{SETH}$, a lower
	bound $2^{\Omega(s)}$ for $\mathcal A$, then $\cc{SETH}$ is false.
\end{lemma}

\begin{proof}
	Write the claimed lower bound as $2^{\gamma s}$ for a constant $\gamma>0$.
	The fine-grained reduction gives constants $\kappa,\delta>0$ such that, on formulas with $N$ variables, all oracle queries have length at most $\kappa N$, and the reduction with a~$2^{(1-\eps)\gamma s}$-time oracle runs in time at most $2^{(1-\delta)N}$ for some fixed $\eps>0$.

	Choose a~small constant $\beta>0$.
	Given an~$n$-variable \problemabbr{SAT} instance, precompute the answers to all $\mathcal A$-instances of length at most $\kappa\beta n$ using the $2^{O(s)}$ algorithm.
	For $\beta$ small enough this table takes time at most $2^{(1-\sigma)n}$ for some $\sigma>0$.
	After that, branch on all but $\beta n$ variables.
	For each of the resulting $2^{(1-\beta)n}$ formulas, run the assumed reduction.
	Its oracle queries have length at most $\kappa\beta n$, so they are answered from the table.
	The total time is
	\[
			2^{(1-\sigma)n} + 2^{(1-\beta)n}\cdot 2^{(1-\delta)\beta n} = 2^{(1-\Omega(1))n},
	\] which contradicts $\cc{SETH}$.
\end{proof}

\bibliographystyle{alpha}
\bibliography{refs}

\end{document}